\documentclass[conference]{IEEEtran}
\usepackage{amsthm}
\usepackage{amsmath}
\usepackage{amssymb}
\usepackage{mathtools}
\usepackage[hidelinks]{hyperref}
\usepackage{cleveref}
\usepackage{stmaryrd}
\usepackage{enumitem}
\usepackage{tikzit}
\usepackage{comment}
\usepackage{eurosym}
\usepackage{booktabs}
\usepackage{etoolbox}
\usepackage[most]{tcolorbox}

\tikzstyle{blank_rectangle}=[fill=white, draw=black, shape=rectangle]
\tikzstyle{new style 0}=[fill=white, draw=black, shape=circle]
\tikzstyle{copy dot}=[fill=black, draw=black, shape=circle]

\tikzstyle{new edge style 0}=[-, fill=white, draw={rgb,255: red,142; green,142; blue,142}]
\tikzstyle{light grey}=[-, draw={rgb,255: red,226; green,226; blue,226}]

\newtheorem{theorem}{Theorem}[section]
\newtheorem{remark}[theorem]{Remark}

\newtheorem{corollary}[theorem]{Corollary}

\theoremstyle{definition}
\newtheorem{definition}[theorem]{Definition}
\theoremstyle{definition}
\newtheorem{notation}[theorem]{Notation}
\theoremstyle{definition}
\newtheorem{example}[theorem]{Example}
\newtheorem{problem}[theorem]{Problem}

\newcommand{\toset}[1]{\mathsf{set}\!\left(#1\right)}     
\newcommand{\inputs}{\mathsf{in}}               
\newcommand{\outputs}{\mathsf{out}}             
\newcommand{\labelfun}{\mathsf{label}}          
\newcommand{\arity}{\mathsf{arity}}             
\newcommand{\ch}{\mathsf{ch}}                   
\newcommand{\listfun}{\mathsf{List}}            
\newcommand{\id}{\mathrm{id}}                   
\newcommand{\swap}{\mathsf{swap}}               
\newcommand{\copymap}{\mathsf{copy}}            
\newcommand{\delmap}{\mathsf{del}}              
\newcommand{\AND}{\mathtt{AND}}                 
\newcommand{\OR}{\mathtt{OR}}                   
\newcommand{\termgraphs}[1]{\mathsf{TermGraphs}_{#1}} 
\newcommand{\sem}[2]{\llbracket #1 \rrbracket_{#2}} 

\newcommand{\mlz}[1]{}

\allowdisplaybreaks

\begin{document}

\title{A Unified Compositional View \\ of Attack Tree Metrics}
\date{}
\author{Benedikt Peterseim and Milan Lopuha\"a-Zwakenberg, \\ \textit{University of Twente}}

\maketitle

\begin{abstract}
    Attack trees (ATs) are popular graphical models for reasoning about the security of complex systems, allowing for the quantification of risk through so-called AT metrics. 
    A large variety of different such AT metrics have been proposed, 
    and despite their wide-spread practical use, 
    no systematic treatment of attack tree metrics so far is fully satisfactory.
    Existing approaches either fail to include important metrics, 
    or they are too general to provide a useful systematic way for defining concrete AT metrics, giving only an abstract characterisation of their behaviour.
    We solve this problem by developing a compositional theory of ATs and their functorial semantics based on gs-monoidal categories. 
    Viewing attack trees as string diagrams, 
    we show that components of ATs form a channel category, a particular type of gs-monoidal category.
    AT metrics then correspond to functors of channel categories.
    This characterisation is both general enough to include all common AT metrics, 
    and concrete enough to define AT metrics by their logical structure.
\end{abstract}

\section{Introduction}\label{sec:introduction}

\noindent \textbf{Background: Attack trees.} 
Since their inception \cite{schneier1999attack}, attack trees (ATs) have been widely applied 
across numerous domains, including  
nuclear control systems \cite{khand2007attack},
smart grids \cite{beckers2014determining},
and railway control systems \cite{dong2017attack}. How they work is best illustrated by an example: 
\Cref{fig:example-attack-tree} shows a simple example of an AT.
Note that an AT is not necessarily a tree in the graph-theoretic sense: some nodes (such as $\mathtt{forge\ badge}$ in \Cref{fig:example-attack-tree}) may have multiple parents.
For a comprehensive overview of ATs in comparison to other graphical models for security modelling, see \cite[Section 3.1.1]{kordy2014dag}.
A number of extensions to the formalism of ATs 
have been proposed, 
such as dynamic ATs \cite{jhawar2015attack} and attack-defence trees \cite{kordy2010foundations}. 
While we focus on standard ATs for simplicity, 
our methods will be set up in a way 
that can easily be adapted to these extensions.

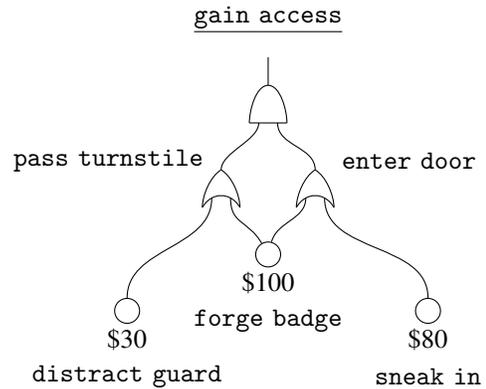
\begin{figure}
        \vspace{-0.4cm}
        \centering
        \begin{tikzpicture}
	\begin{pgfonlayer}{nodelayer}
		\node [style=new style 0] (0) at (-3.5, -1.75) {};
		\node [style=new style 0] (3) at (4.5, -1.75) {};
		\node [style=new style 0] (9) at (0.25, -0.25) {};
		\node [style=none] (12) at (-1, 2) {};
		\node [style=none] (15) at (1.5, 2) {};
		\node [style=none] (18) at (0.25, 4.25) {};
		\node [style=none] (19) at (-1.25, 1.25) {};
		\node [style=none] (23) at (1.25, 1.25) {};
		\node [style=none] (24) at (1.75, 1.25) {};
		\node [style=none] (26) at (-0.75, 1.25) {};
		\node [style=none] (27) at (0, 3.25) {};
		\node [style=none] (28) at (0.5, 3.25) {};
		\node [style=none] (29) at (0.25, 5) {};
		\node [style=none] (39) at (-3.5, -2.5) {\$30};
		\node [style=none] (40) at (0.25, -1) {\$100};
		\node [style=none] (44) at (4.5, -2.5) {\$80};
		\node [style=none] (54) at (0.25, 6) {\underline{$\mathtt{gain\ access}$}};
		\node [style=none] (56) at (4, 2.25) {$\mathtt{enter\ door}$};
		\node [style=none] (59) at (4.5, -3.5) {$\mathtt{sneak\ in}$};
		\node [style=none] (60) at (-4, 2.25) {$\mathtt{pass\ turnstile}$};
		\node [style=none] (78) at (0.25, 4.25) {};
		\node [style=none] (79) at (0, 3.25) {};
		\node [style=none] (80) at (0.5, 3.25) {};
		\node [style=none] (81) at (-0.25, 3.25) {};
		\node [style=none] (82) at (0.75, 3.25) {};
		\node [style=none] (83) at (0.25, 4.25) {};
		\node [style=none] (84) at (0, 3.25) {};
		\node [style=none] (85) at (0.5, 3.25) {};
		\node [style=none] (86) at (1.5, 2) {};
		\node [style=none] (87) at (1.25, 1.25) {};
		\node [style=none] (88) at (1.75, 1.25) {};
		\node [style=none] (89) at (1, 1) {};
		\node [style=none] (90) at (2, 1) {};
		\node [style=none] (91) at (1.5, 2) {};
		\node [style=none] (92) at (1.25, 1.25) {};
		\node [style=none] (93) at (1.75, 1.25) {};
		\node [style=none] (94) at (-1, 2) {};
		\node [style=none] (95) at (-1.25, 1.25) {};
		\node [style=none] (96) at (-0.75, 1.25) {};
		\node [style=none] (97) at (-1.5, 1) {};
		\node [style=none] (98) at (-0.5, 1) {};
		\node [style=none] (99) at (-1, 2) {};
		\node [style=none] (100) at (-1.25, 1.25) {};
		\node [style=none] (101) at (-0.75, 1.25) {};
		\node [style=none] (102) at (-3.5, -3.5) {$\mathtt{distract\ guard}$};
		\node [style=none] (104) at (0.25, -2) {$\mathtt{forge\ badge}$};
		\node [style=none] (107) at (0, -4.25) {};
	\end{pgfonlayer}
	\begin{pgfonlayer}{edgelayer}
		\draw [in=-90, out=90] (0) to (19.center);
		\draw [in=90, out=-90] (24.center) to (3);
		\draw [in=75, out=-90] (23.center) to (9);
		\draw [in=120, out=-90, looseness=1.25] (26.center) to (9);
		\draw [in=-90, out=90, looseness=0.75] (12.center) to (27.center);
		\draw [in=-90, out=90] (15.center) to (28.center);
		\draw (18.center) to (29.center);
		\draw [in=90, out=-180, looseness=0.75] (83.center) to (81.center);
		\draw [in=90, out=0, looseness=0.75] (83.center) to (82.center);
		\draw (81.center) to (82.center);
		\draw [in=90, out=-150, looseness=0.75] (91.center) to (89.center);
		\draw [in=90, out=-30, looseness=0.75] (91.center) to (90.center);
		\draw [bend left=60, looseness=1.25] (89.center) to (90.center);
		\draw [in=90, out=-150, looseness=0.75] (99.center) to (97.center);
		\draw [in=90, out=-30, looseness=0.75] (99.center) to (98.center);
		\draw [bend left=60, looseness=1.25] (97.center) to (98.center);
	\end{pgfonlayer}
\end{tikzpicture}
        \caption{
        \textit{Example attack tree:} 
        An attacker aims to gain unauthorised access 
        to an office space. 
        This requires passing a turnstile 
        and entering through a secured door 
        (top AND gate). 
        To pass the turnstile, 
        they can either forge an access badge 
        or distract the guard and jump over 
        (left OR gate). 
        To enter the door, 
        they can either also use a forged badge 
        or follow an employee in 
        (right OR gate). 
        Basic attack steps are labelled with their cost: 
        distracting the guard costs $\$30$, forging a badge $\$100$, and sneaking in $\$80$.
        } 
        \vspace{-0.5cm}
        \label{fig:example-attack-tree}
\end{figure}

\text{}\\
\noindent \textbf{Attack tree metrics.} 
In addition to a \emph{qualitative} model of how basic attack steps combine via AND and OR gates, 
ATs allow for the quantification of risk through AT metrics. 
For example, in \Cref{fig:example-attack-tree}, 
we use the \emph{min.~cost metric},
associating to each basic attack step the minimum cost required to successfully perform it.
The question of computing the \emph{metric value} is now: 
how do these attribute values propagate from the basic attack steps to the top node?
In our example: what is the minimum cost required to successfully gain access to the office space? 

Already in the example of the min.~cost metric, this question may be answered in different ways: 
depending on how one interprets the AT, 
the minimum cost of a successful attack may either be 
$\$110$ or $\$100$ for the AT in \Cref{fig:example-attack-tree}; see \Cref{sec:two-interpretations}.
In addition, many other types of metrics are used, 
including the minimum time or skill level required for a successful attack, 
or its maximal probability (under certain assumptions on attacker behaviour). 
With each of the various ways to quantify risk in an AT admitting several interpretations as an AT metric, the need for a systematic unified account of AT metrics emerges:

\begin{tcolorbox}
\begin{problem}\label{main-problem}
Provide a framework for attack tree metrics that is both:
\begin{enumerate}
    \item sufficiently \emph{general} to capture all common attack tree metrics used in practice, and
    \item sufficiently \emph{concrete} for attack tree metrics to be fully determined by how they behave on 
    the basic building blocks of attack trees: 
    logic gates and basic attack steps. 
\end{enumerate}
\end{problem}
\end{tcolorbox}
\noindent The second property ensures that a framework for AT metrics 
yields a systematic method for defining AT metrics, 
instead of just an abstract characterisation of their general behaviour.

\text{}\\
\noindent\textbf{Shortcomings of existing frameworks.}
Despite the widespread use of ATs,
so far, 
no systematic treatment of AT metrics exists that meets both of the requirements of \Cref{main-problem}.
    
Several mutually incompatible frameworks based on \emph{semirings} (see \Cref{def:semiring}) exist,
one based on what we will refer to as the \emph{bottom-up interpretation} \cite{mauw2005foundations}, 
one based on the \emph{propositional interpretation} \cite{lopuhaa2022efficient},
and another one based on yet a different interpretation \cite{bossuat2017evil}. 
All are concrete enough to provide a recipe for defining AT metrics as required 
by our second desideratum. 
However, they do not satisfy our first requirement, failing to capture those AT metrics 
that do not adhere to their chosen interpretation.
    
In the other direction, earlier work based on \emph{operads} \cite{lopuhaa2024attack}
provided a framework that does include essentially every important AT metric, 
instead failing to satisfy our second requirement of sufficient concreteness. 
In particular, specifying an AT metric as defined in this framework requires specifying its value on infinitely many ATs, 
which is clearly not satisfactory.

\text{}\\
\textbf{Approach.} We introduce a new framework for AT metrics that solves Problem \ref{main-problem}. The key idea is to consider building blocks of ATs with multiple inputs and outputs. 
We define a metric by the operations it associates to AND/OR gates, and how it handles the dependence between multiple outputs. 
This is general enough to capture all relevant metrics, 
and concrete enough to specify metrics from finite data.
We differ from semiring formalisms (which only describe the AND/OR operators and hardcode the dependency handling) 
and from the operad approach (which needs to specify infinitely many operators, so it cannot be used computationally in full generality).

Mathematically, our framework is based on \emph{gs-monoidal categories}, a category-theoretic concept originally proposed in the context of term rewriting \cite{gadducci1996algebraic}. Nevertheless, we will assume no prior knowledge of category theory, as we give a self-contained definition of a special type of gs-monoidal category we call \emph{channel categories} for their resemblance to information-theoretic channels. A channel category is defined in terms of abstract operations of parallel and sequential composition: how a metric handles these operations exactly describes how a metric handles dependencies within ATs. 
We show that attack trees give rise to a channel category, and that the same holds for most metrics from the literature (with only a few exceptions, discussed in \Cref{sec:non-examples}). 
Thus, we define an AT metric as a \emph{functor of channel categories}. 
Despite its generality, this definition is still sufficiently concrete to yield a method for computing the metric.
Interestingly, the different AT semantics that have been proposed in the literature \cite{mauw2005foundations,kordy2014dag,lopuhaa2022efficient} can also be understood as functors of channel categories; 
hence our approach unifies AT metrics and semantics. 

\text{}\\
\noindent\textbf{Contributions.} In summary, our main contributions are:
\begin{itemize}
\item A new, compositional framework for AT metrics;
\item A self-contained exposition of the underlying mathematical theory of channel categories;
\item Showing most existing metrics fit into this framework;
\item Unifying AT metrics and semantics.
\end{itemize}

\noindent Understanding the compositionality of ATs is also of independent interest,
beyond the problem of systematising AT metrics.

First, it provides a more formal view on how modelling with ATs is done \emph{compositionally},
by building up models of complex systems from simpler ones.
In the context of \emph{fault trees} (the counterpart of ATs in reliability engineering), 
\emph{component fault trees} already provide such formalism for compositional modelling \cite{kaiser2003new}.
Our formalism is more general, 
since fault trees can be viewed as a special case of ATs in our setting 
(where the ``attacker'' is nature).

Moreover, our treatment of ATs as \emph{term graphs} (see \Cref{sec:term-graphs}) presents them as particular kinds of \emph{string diagrams},
fitting into a wider picture.
String diagrams have broadly been applied for a compositional understanding 
of graphical languages across various domains: Bayesian networks \cite{fong2013causal,jacobs2018channel}, 
quantum circuits \cite{coecke2018picturing}, 
and tensor networks \cite{biamonte2011categorical}.
``String diagram'' is an umbrella term used for various related graphical calculi, 
and term graphs are the specific form of string diagram corresponding to the context of channel categories.

Finally, compositional semantics are an inherently desirable property of any formal language,
ensuring that the meaning of complex terms is systematically built up from the meanings of simpler terms.
Our results can then also be interpreted as providing \textbf{functorial (or ``categorical'') semantics} for ATs.

\section{Example: Two Min.~Cost Metrics}\label{sec:example}

\noindent In this section, 
we highlight the key aspects of our approach through a concrete example. 
We show how two seemingly incompatible interpretations of the min.~cost metric, the \emph{bottom-up} \cite{mauw2005foundations} and the \emph{propositional} \cite{lopuhaa2022efficient} interpretation,
can be defined in a unified way within our framework. 

\subsection{Two interpretations}\label{sec:two-interpretations} 

\textit{The bottom-up interpretation \cite{mauw2005foundations}.} In this interpretation, 
the AND and OR gates directly denote operations that combine metric values.
In the case of minimum cost, 
an AND gate is interpreted as taking the sum of the costs of its children,
since \emph{all} sub-goals must be achieved to successfully execute the attack associated to the AND gate.
OR gates are interpreted as taking the minimum of the costs of its children,
as only \emph{some} sub-goal needs to be achieved for the OR gate to be activated.
Hence, in the example AT from \Cref{fig:example-attack-tree}, the minimum cost of a successful attack 
under this bottom-up interpretation is: 
    $$ \min(\$30, \$100) + \min(\$100, \$80) = \$110. $$
Importantly, under the bottom-up interpretation, the basic attack step labelled 
$\mathtt{forge\ badge}$ needs to be executed separately in each sub-goal 
in which it appears. 
That is, the attacker cannot use the same access badge
for both passing the turnstile and for entering through the office door. 
If instead we would like to model the same forged access badge being used in both sub-goals,
we either need a different attack tree, or we may use the \emph{propositional} interpretation. 

\textit{The propositional interpretation \cite{lopuhaa2022efficient}.} 
This interpretation 
views an attack tree as a propositional formula, 
describing which sets of basic attack steps lead to achieving the top-level goal. 
Such set of basic attack steps is called 
a \emph{successful attack}.
The minimum cost is computed by considering all
\emph{minimal} successful attacks, 
those successful attacks that have no proper successful subsets.
For the example attack tree from \Cref{fig:example-attack-tree}, there are two minimal successful attacks: 
one where the attacker performs only the basic attack step labelled  $\mathtt{forge\ badge}$; 
and one attack where the steps labelled $\mathtt{distract\ guard}$ and $\mathtt{sneak\ in}$ (but not $\mathtt{forge\ badge}$) are performed. 
These minimal successful attacks have costs $\$100$ and $\$30+\$80$, respectively. 
Therefore, the minimum cost  is 
    $$ \min(\$100, \$30 + \$80) = \$100. $$
Since the attacker can use the same forged access badge for both passing the turnstile 
and entering through the office door,
the minimum cost of a successful attack is strictly lower than in the bottom-up interpretation.
In particular, the two interpretations are not equivalent. 

\subsection{A unified way to define both attack tree metrics}

The reason why these two interpretations of min.~cost yield different values lies in how they treat \emph{dependence}. 
To make this more mathematically precise, consider the ``multiple-root attack tree’’ resulting from removing the top-level AND-gate of the example attack tree in \Cref{fig:example-attack-tree}. 
The two top-level goals in this example are $\mathtt{pass\ turnstile}$ and $\mathtt{enter\ door}$: \mlz{Put BAS costs in this figure}
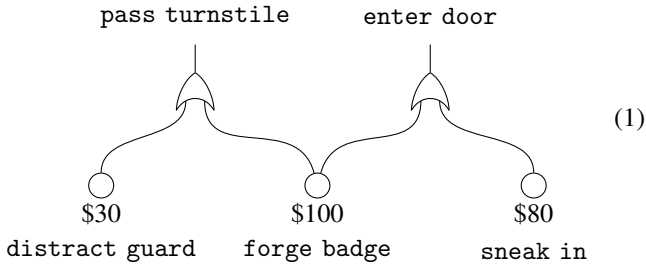
\begin{equation}\label{eq:multi-root-attack-tree}
    \begin{tikzpicture}
	\begin{pgfonlayer}{nodelayer}
		\node [style=new style 0] (111) at (-5.75, -1.75) {};
		\node [style=new style 0] (112) at (5.75, -1.75) {};
		\node [style=new style 0] (113) at (0, -1.75) {};
		\node [style=none] (114) at (-3.25, 1.25) {};
		\node [style=none] (115) at (3, 1.25) {};
		\node [style=none] (117) at (-3.5, 0.5) {};
		\node [style=none] (118) at (2.75, 0.5) {};
		\node [style=none] (119) at (3.25, 0.5) {};
		\node [style=none] (120) at (-3, 0.5) {};
		\node [style=none] (121) at (-3.25, 2) {};
		\node [style=none] (122) at (3, 2) {};
		\node [style=none] (136) at (3, 1.25) {};
		\node [style=none] (137) at (2.75, 0.5) {};
		\node [style=none] (138) at (3.25, 0.5) {};
		\node [style=none] (139) at (2.5, 0.25) {};
		\node [style=none] (140) at (3.5, 0.25) {};
		\node [style=none] (141) at (3, 1.25) {};
		\node [style=none] (142) at (2.75, 0.5) {};
		\node [style=none] (143) at (3.25, 0.5) {};
		\node [style=none] (144) at (-3.25, 1.25) {};
		\node [style=none] (145) at (-3.5, 0.5) {};
		\node [style=none] (146) at (-3, 0.5) {};
		\node [style=none] (147) at (-3.75, 0.25) {};
		\node [style=none] (148) at (-2.75, 0.25) {};
		\node [style=none] (149) at (-3.25, 1.25) {};
		\node [style=none] (150) at (-3.5, 0.5) {};
		\node [style=none] (151) at (-3, 0.5) {};
		\node [style=none] (152) at (5.75, -3.5) {$\mathtt{sneak\ in}$};
		\node [style=none] (153) at (-5.75, -3.5) {$\mathtt{distract\ guard}$};
		\node [style=none] (154) at (0, -3.5) {$\mathtt{forge\ badge}$};
		\node [style=none] (155) at (3, 2.75) {$\mathtt{enter\ door}$};
		\node [style=none] (156) at (-3.25, 2.75) {$\mathtt{pass\ turnstile}$};
		\node [style=none] (157) at (-5.75, -2.5) {\$30};
		\node [style=none] (158) at (0, -2.5) {\$100};
		\node [style=none] (159) at (5.75, -2.5) {\$80};
		\node [style=none] (160) at (0, -4) {};
		\node [style=none] (161) at (0, 3.25) {};
	\end{pgfonlayer}
	\begin{pgfonlayer}{edgelayer}
		\draw [in=-90, out=90] (111) to (117.center);
		\draw [in=90, out=-90] (119.center) to (112);
		\draw [in=75, out=-90, looseness=1.25] (118.center) to (113);
		\draw [in=105, out=-90, looseness=1.25] (120.center) to (113);
		\draw [in=-90, out=90, looseness=1.25] (114.center) to (121.center);
		\draw [in=-90, out=90, looseness=1.25] (115.center) to (122.center);
		\draw [in=90, out=-150, looseness=0.75] (141.center) to (139.center);
		\draw [in=90, out=-30, looseness=0.75] (141.center) to (140.center);
		\draw [bend left=60, looseness=1.25] (139.center) to (140.center);
		\draw [in=90, out=-150, looseness=0.75] (149.center) to (147.center);
		\draw [in=90, out=-30, looseness=0.75] (149.center) to (148.center);
		\draw [bend left=60, looseness=1.25] (147.center) to (148.center);
	\end{pgfonlayer}
\end{tikzpicture}
\end{equation}
In the bottom-up interpretation, evaluating a multiple-root attack tree is straightforward: we compute the cost of each output independently, giving output vector $\begin{pmatrix} 80 & 30 \end{pmatrix}^{\intercal}$ (reading outputs from left to right). Taking their sum, the metric value of the larger AT can be computed from these values alone.

In contrast, simply storing the outputs $\begin{pmatrix} 80 & 30 \end{pmatrix}^{\intercal}$ is not enough for the propositional interpretation, as it does not tell us that we can activate both top gates more cheaply by the single basic attack step labelled $\mathtt{forge\ badge}$. 
To account for this, our output is a four-dimensional vector, whose entries correspond to the minimum cost required to obtain top-level outputs $00$, $01$, $10$, $11$; the result is $\begin{pmatrix} 0 & 80 & 30 & 100 \end{pmatrix}^{\intercal}$. 
This is all the information later gates need to compute min.~costs according to the propositional interpretation.  
AND/OR gates act on such vectors as matrices over the $(\min,+)$-semiring. 
In fact, the top AND gate becomes the matrix $\bigl(\begin{smallmatrix} 0 & 0 & 0 & \infty \\ \infty & \infty & \infty & 0 \end{smallmatrix}\bigr)$, so the result is:
\begin{gather*}
\begin{pmatrix} 0 & 0 & 0 & \infty  \\ \infty & \infty & \infty & 0\end{pmatrix} \begin{pmatrix} 0 \\80\\30\\100 \end{pmatrix}\\
= \begin{pmatrix}\min(0+0,0+80,0+30,\infty+100)\\\min(\infty+0,\infty+80,\infty+30,0+100) \end{pmatrix} 
= \begin{pmatrix}0 \\ 100\end{pmatrix}.
\end{gather*}
The first coefficient is the minimum cost to get output 0 for the overall AT, and the second coefficient the minimal cost to get output 1; this is the $\$100$ we computed before.

This perspective allows us to formalise both interpretations as mappings of the form:
\[
\begin{aligned}
\sem{\,\cdot\,}{\mathsf{bottom}\text{-}\mathsf{up}}\colon\;\;\,\mathsf{AttackTrees} &\longrightarrow \mathsf{Functions}, \\
\sem{\,\cdot\,}{\mathsf{propositional}}\colon\mathsf{AttackTrees} &\longrightarrow \mathsf{Matrices},
\end{aligned}
\]
where \textsf{AttackTrees} denotes a certain kind of category of attack trees, and \textsf{Functions} and \textsf{Matrices} are categories of the same kind capturing the semantics of the bottom-up and propositional interpretations, respectively.

More precisely, these mappings are \emph{functors of channel categories}. Informally, a channel category consists of sets of ``channels’’ (such as functions or matrices), 
together with operations for sequential and parallel composition 
of these channels, and a number of special ``wiring’’ channels.

To understand attack trees in a compatible way, we need to pass to \emph{attack tree components}. 
Attack tree components are parts of an attack tree, cut out from $T$ along edges:
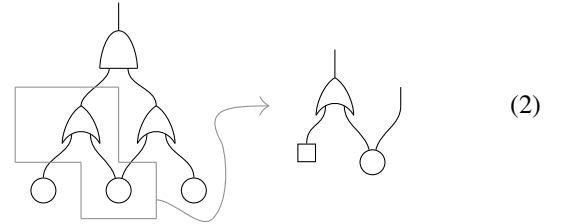
\begin{equation}\label{eq:attack-tree-component-intro}
    \begin{tikzpicture}
	\begin{pgfonlayer}{nodelayer}
		\node [style=new style 0] (0) at (-5.75, -2.25) {};
		\node [style=new style 0] (3) at (-1.75, -2.25) {};
		\node [style=new style 0] (9) at (-3.75, -2.25) {};
		\node [style=none] (12) at (-4.75, 0) {};
		\node [style=none] (15) at (-2.75, 0) {};
		\node [style=none] (18) at (-3.75, 2) {};
		\node [style=none] (19) at (-5, -0.75) {};
		\node [style=none] (23) at (-3, -0.75) {};
		\node [style=none] (24) at (-2.5, -0.75) {};
		\node [style=none] (26) at (-4.5, -0.75) {};
		\node [style=none] (27) at (-4, 1) {};
		\node [style=none] (28) at (-3.5, 1) {};
		\node [style=none] (29) at (-3.75, 2.75) {};
		\node [style=none] (78) at (-3.75, 2) {};
		\node [style=none] (79) at (-4, 1) {};
		\node [style=none] (80) at (-3.5, 1) {};
		\node [style=none] (81) at (-4.25, 1) {};
		\node [style=none] (82) at (-3.25, 1) {};
		\node [style=none] (83) at (-3.75, 2) {};
		\node [style=none] (84) at (-4, 1) {};
		\node [style=none] (85) at (-3.5, 1) {};
		\node [style=none] (86) at (-2.75, 0) {};
		\node [style=none] (87) at (-3, -0.75) {};
		\node [style=none] (88) at (-2.5, -0.75) {};
		\node [style=none] (89) at (-3.25, -1) {};
		\node [style=none] (90) at (-2.25, -1) {};
		\node [style=none] (91) at (-2.75, 0) {};
		\node [style=none] (92) at (-3, -0.75) {};
		\node [style=none] (93) at (-2.5, -0.75) {};
		\node [style=none] (94) at (-4.75, 0) {};
		\node [style=none] (95) at (-5, -0.75) {};
		\node [style=none] (96) at (-4.5, -0.75) {};
		\node [style=none] (97) at (-5.25, -1) {};
		\node [style=none] (98) at (-4.25, -1) {};
		\node [style=none] (99) at (-4.75, 0) {};
		\node [style=none] (100) at (-5, -0.75) {};
		\node [style=none] (101) at (-4.5, -0.75) {};
		\node [style=none] (102) at (-6.5, -1.5) {};
		\node [style=none] (103) at (-6.5, 0.5) {};
		\node [style=none] (104) at (-3.75, 0.5) {};
		\node [style=none] (105) at (-4.75, -1.5) {};
		\node [style=none] (106) at (-4.75, -3) {};
		\node [style=none] (107) at (-2.75, -3) {};
		\node [style=none] (108) at (-2.75, -1.5) {};
		\node [style=none] (109) at (-3.75, -1.5) {};
		\node [style=none] (110) at (-2.75, -2.5) {};
		\node [style=none] (111) at (-1, -1) {};
		\node [style=none] (112) at (0.25, 0) {};
		\node [style=none] (113) at (-1.25, -3) {};
		\node [style=none] (114) at (0, 0.25) {};
		\node [style=none] (115) at (0, -0.25) {};
		\node [style={blank_rectangle}] (116) at (1.25, -1.25) {};
		\node [style=new style 0] (117) at (3, -1.5) {};
		\node [style=none] (118) at (2, 0.75) {};
		\node [style=none] (120) at (1.75, 0) {};
		\node [style=none] (121) at (3.75, 0) {};
		\node [style=none] (123) at (2.25, 0) {};
		\node [style=none] (124) at (2, 1.5) {};
		\node [style=none] (140) at (2, 0.75) {};
		\node [style=none] (141) at (1.75, 0) {};
		\node [style=none] (142) at (2.25, 0) {};
		\node [style=none] (143) at (1.5, -0.25) {};
		\node [style=none] (144) at (2.5, -0.25) {};
		\node [style=none] (145) at (2, 0.75) {};
		\node [style=none] (146) at (1.75, 0) {};
		\node [style=none] (147) at (2.25, 0) {};
		\node [style=none] (148) at (3.75, 0.5) {};
	\end{pgfonlayer}
	\begin{pgfonlayer}{edgelayer}
		\draw [in=-90, out=90] (0) to (19.center);
		\draw [in=90, out=-90] (24.center) to (3);
		\draw [in=75, out=-90] (23.center) to (9);
		\draw [in=105, out=-90] (26.center) to (9);
		\draw [in=-90, out=90, looseness=0.75] (12.center) to (27.center);
		\draw [in=-90, out=90] (15.center) to (28.center);
		\draw (18.center) to (29.center);
		\draw [in=90, out=-180, looseness=0.75] (83.center) to (81.center);
		\draw [in=90, out=0, looseness=0.75] (83.center) to (82.center);
		\draw (81.center) to (82.center);
		\draw [in=90, out=-150, looseness=0.75] (91.center) to (89.center);
		\draw [in=90, out=-30, looseness=0.75] (91.center) to (90.center);
		\draw [bend left=60, looseness=1.25] (89.center) to (90.center);
		\draw [in=90, out=-150, looseness=0.75] (99.center) to (97.center);
		\draw [in=90, out=-30, looseness=0.75] (99.center) to (98.center);
		\draw [bend left=60, looseness=1.25] (97.center) to (98.center);
		\draw [style=new edge style 0] (105.center) to (106.center);
		\draw [style=new edge style 0] (106.center) to (107.center);
		\draw [style=new edge style 0] (108.center) to (107.center);
		\draw [style=new edge style 0] (105.center) to (102.center);
		\draw [style=new edge style 0] (103.center) to (102.center);
		\draw [style=new edge style 0] (103.center) to (104.center);
		\draw [style=new edge style 0] (104.center) to (109.center);
		\draw [style=new edge style 0] (109.center) to (108.center);
		\draw [style=new edge style 0, in=-135, out=0, looseness=0.75] (110.center) to (113.center);
		\draw [style=new edge style 0, in=-75, out=45, looseness=0.75] (113.center) to (111.center);
		\draw [style=new edge style 0, in=180, out=120, looseness=1.75] (111.center) to (112.center);
		\draw [style=new edge style 0, bend right=15] (114.center) to (112.center);
		\draw [style=new edge style 0, bend left=15, looseness=1.25] (115.center) to (112.center);
		\draw [in=-90, out=90, looseness=1.25] (116) to (120.center);
		\draw [in=75, out=-90] (121.center) to (117);
		\draw [in=105, out=-90] (123.center) to (117);
		\draw [in=-90, out=90, looseness=0.75] (118.center) to (124.center);
		\draw [in=90, out=-150, looseness=0.75] (145.center) to (143.center);
		\draw [in=90, out=-30, looseness=0.75] (145.center) to (144.center);
		\draw [bend left=60, looseness=1.25] (143.center) to (144.center);
		\draw (121.center) to (148.center);
	\end{pgfonlayer}
\end{tikzpicture}
\end{equation}
In contrast to a full attack tree, an attack tree component may have 
multiple inputs, denoted by small squares (``$\square$''), 
and outputs, represented by outgoing wires. 
For example, the attack tree component on the right hand side of \eqref{eq:attack-tree-component-intro}
has one input and two outputs; the multiple-root attack tree \eqref{eq:multi-root-attack-tree}
has two outputs and no input. \Cref{sec:term-graphs} will make attack tree components precise, by defining them as certain \emph{term graphs}.
Channel categories will be introduced in detail in \Cref{sec:channels}.

The perspective of attack tree metrics as functors of channel categories accounts not only for the two instantiations of the min.~cost metric discussed here, but also for a wide range of metrics and semantics found in the literature. This is the subject of \Cref{sec:compositionality-of-attack-tree-metrics}.

\section{Syntax: Term Graphs}\label{sec:term-graphs}

On a syntactic level, 
graphical risk models such as 
attack trees, fault trees and their numerous variants 
can all be understood as \emph{term graphs}. 
Term graphs 
can be viewed as the single-sorted (``untyped'') special case of \emph{gs-monoidal string diagrams}  \cite{corradini1999algebraic,fritz2023free,cho2019disintegration},  
but we will not assume any familiarity with string diagrams or category theory.
Alternatively, term graphs can be viewed as a generalisation of syntax trees, 
allowing for explicit sharing between subtrees.

Our main focus will be on particular term graphs: attack tree components.
Nevertheless, we work in the general setting of term graphs for two reasons.
First, it allows us to directly apply the results of \cite{corradini1999algebraic} to our case of interest.
Second, it provides a basis to easily generalise our results to more general variants of attack trees. 

Term graphs are always defined with respect to a certain \emph{signature}.

\begin{definition}[Signature]
    A (single-sorted, algebraic) \emph{signature} is a set $\Sigma$ 
    together with a function $\arity_\Sigma: \Sigma \to \mathbb{N}$.
    The elements of $\Sigma$ are called \emph{function symbols} and each $f\in \Sigma$ 
    is thought of as taking $\arity_\Sigma(f)$ many arguments.
\end{definition}

\begin{example}
    A typical example of a signature in the context of algebra is 
    the signature of groups,
    \begin{gather*}
    \mathsf{Gr} := \{m, i, 1\},\\
    \arity_{\mathsf{Gr}}(m) = 2, 
        \; \arity_{\mathsf{Gr}}(i) = 1, 
        \; \arity_{\mathsf{Gr}}(1) = 0.
    \end{gather*}
    Here, $m$ denotes the group operation, or multiplication (a binary operation);
    $i$ the inversion (a unary operation); 
    and $1$ the identity element (a nullary operation, or constant).
\end{example}

Our signature of interest is the following.

\begin{definition}\label{ex:attack-and-fault-tree-signatures}
    Let $B$ be a set, thought of as consisting of \emph{basic attack step labels}.
    The \emph{signature of attack trees over $B$} is 
    \begin{gather*}
        \mathsf{AT}(B) := \{\AND_i, \OR_i \mid i\in \mathbb{N}_{\geq 1}\} \sqcup B, \\
        \arity_{\mathsf{AT}(B)}(\AND_i) = \arity_{\mathsf{AT}(B)}(\OR_i) = i, \\
        \arity_{\mathsf{AT}(B)}(b) = 0,
    \end{gather*}
    for all $i\in \mathbb{N}_{\geq 1}, b\in B$. Here, $\sqcup$ denotes the disjoint union of sets.
\end{definition}

For each $i\in \mathbb{N}_{\geq 1}$, the function symbols $\AND_i$ and $\OR_i$ are thought of as representing an AND gate and an OR gate with $i$ inputs, respectively.

\begin{example}\label{ex:at-signature-example}
    The AT from \Cref{fig:example-attack-tree} is defined over the signature $\mathsf{AT}(B)$ 
    with basic attack steps labels
        $$ B := \{ \mathtt{distract\ guard}, \mathtt{forge\ badge}, \mathtt{sneak\ in} \}.$$
    We will abbreviate $B = \{\mathtt{D},\mathtt{F},\mathtt{S}\}$ from now on.
\end{example}
\mlz{Change W to S in figures.}

At this point, 
one might be tempted to expect that term graphs are simply graph representations 
of terms over the given signature. 
This is \emph{not} the case, as term graphs allow for explicit sharing of information. 
For example, as a consequence of the discussion in \Cref{sec:example}, we see that the following 
two attack trees are different under the propositional interpretation:
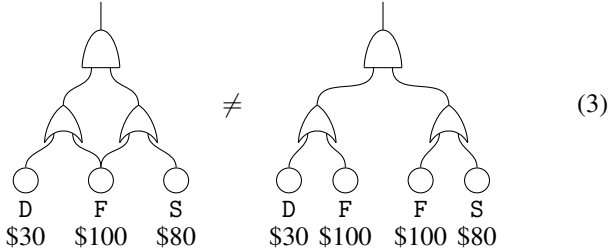
\begin{equation}
    \begin{tikzpicture}
	\begin{pgfonlayer}{nodelayer}
		\node [style=none] (27) at (-0.75, 0) {$\not=$};
		\node [style=new style 0] (111) at (0.75, -2) {};
		\node [style=new style 0] (112) at (5.75, -2) {};
		\node [style=new style 0] (113) at (2.25, -2) {};
		\node [style=none] (114) at (1.5, 0) {};
		\node [style=none] (115) at (5, 0) {};
		\node [style=none] (116) at (3.25, 2) {};
		\node [style=none] (117) at (1.25, -0.75) {};
		\node [style=none] (119) at (5.25, -0.75) {};
		\node [style=none] (120) at (1.75, -0.75) {};
		\node [style=none] (121) at (3, 1) {};
		\node [style=none] (122) at (3.5, 1) {};
		\node [style=none] (123) at (3.25, 2.75) {};
		\node [style=none] (128) at (3.25, 2) {};
		\node [style=none] (129) at (3, 1) {};
		\node [style=none] (130) at (3.5, 1) {};
		\node [style=none] (131) at (2.75, 1) {};
		\node [style=none] (132) at (3.75, 1) {};
		\node [style=none] (133) at (3.25, 2) {};
		\node [style=none] (134) at (3, 1) {};
		\node [style=none] (135) at (3.5, 1) {};
		\node [style=none] (136) at (5, 0) {};
		\node [style=none] (138) at (5.25, -0.75) {};
		\node [style=none] (139) at (4.5, -1) {};
		\node [style=none] (140) at (5.5, -1) {};
		\node [style=none] (141) at (5, 0) {};
		\node [style=none] (143) at (5.25, -0.75) {};
		\node [style=none] (144) at (1.5, 0) {};
		\node [style=none] (145) at (1.25, -0.75) {};
		\node [style=none] (146) at (1.75, -0.75) {};
		\node [style=none] (147) at (1, -1) {};
		\node [style=none] (148) at (2, -1) {};
		\node [style=none] (149) at (1.5, 0) {};
		\node [style=none] (150) at (1.25, -0.75) {};
		\node [style=none] (151) at (1.75, -0.75) {};
		\node [style=new style 0] (152) at (-6.25, -2) {};
		\node [style=new style 0] (153) at (-2.25, -2) {};
		\node [style=new style 0] (154) at (-4.25, -2) {};
		\node [style=none] (155) at (-5.25, 0) {};
		\node [style=none] (156) at (-3.25, 0) {};
		\node [style=none] (157) at (-4.25, 2) {};
		\node [style=none] (158) at (-5.5, -0.75) {};
		\node [style=none] (159) at (-3.5, -0.75) {};
		\node [style=none] (160) at (-3, -0.75) {};
		\node [style=none] (161) at (-5, -0.75) {};
		\node [style=none] (162) at (-4.5, 1) {};
		\node [style=none] (163) at (-4, 1) {};
		\node [style=none] (164) at (-4.25, 2.75) {};
		\node [style=none] (165) at (-4.25, 2) {};
		\node [style=none] (166) at (-4.5, 1) {};
		\node [style=none] (167) at (-4, 1) {};
		\node [style=none] (168) at (-4.75, 1) {};
		\node [style=none] (169) at (-3.75, 1) {};
		\node [style=none] (170) at (-4.25, 2) {};
		\node [style=none] (171) at (-4.5, 1) {};
		\node [style=none] (172) at (-4, 1) {};
		\node [style=none] (173) at (-3.25, 0) {};
		\node [style=none] (174) at (-3.5, -0.75) {};
		\node [style=none] (175) at (-3, -0.75) {};
		\node [style=none] (176) at (-3.75, -1) {};
		\node [style=none] (177) at (-2.75, -1) {};
		\node [style=none] (178) at (-3.25, 0) {};
		\node [style=none] (179) at (-3.5, -0.75) {};
		\node [style=none] (180) at (-3, -0.75) {};
		\node [style=none] (181) at (-5.25, 0) {};
		\node [style=none] (182) at (-5.5, -0.75) {};
		\node [style=none] (183) at (-5, -0.75) {};
		\node [style=none] (184) at (-5.75, -1) {};
		\node [style=none] (185) at (-4.75, -1) {};
		\node [style=none] (186) at (-5.25, 0) {};
		\node [style=none] (187) at (-5.5, -0.75) {};
		\node [style=none] (188) at (-5, -0.75) {};
		\node [style=new style 0] (189) at (4.25, -2) {};
		\node [style=none] (190) at (4.75, -0.75) {};
		\node [style=none] (194) at (0.75, -3.5) {};
		\node [style=none] (195) at (2.25, -3.5) {\$100};
		\node [style=none] (196) at (4.25, -3.5) {\$100};
		\node [style=none] (197) at (5.75, -3.5) {\$80};
		\node [style=none] (198) at (-6.25, -3.5) {\$30};
		\node [style=none] (199) at (-4.25, -3.5) {\$100};
		\node [style=none] (200) at (-2.25, -3.5) {\$80};
		\node [style=none] (201) at (0.75, -3.5) {\$30};
		\node [style=none] (202) at (0.75, -2.75) {};
		\node [style=none] (203) at (2.25, -2.75) {$\mathtt{F}$};
		\node [style=none] (204) at (4.25, -2.75) {$\mathtt{F}$};
		\node [style=none] (205) at (5.75, -2.75) {$\mathtt{S}$};
		\node [style=none] (206) at (-6.25, -2.75) {$\mathtt{D}$};
		\node [style=none] (207) at (-4.25, -2.75) {$\mathtt{F}$};
		\node [style=none] (208) at (-2.25, -2.75) {$\mathtt{S}$};
		\node [style=none] (209) at (0.75, -2.75) {$\mathtt{D}$};
		\node [style=none] (210) at (7.5, 0) {};
		\node [style=none] (211) at (0, 3) {};
		\node [style=none] (212) at (0, -4) {};
	\end{pgfonlayer}
	\begin{pgfonlayer}{edgelayer}
		\draw [in=-90, out=90] (111) to (117.center);
		\draw [in=90, out=-90] (119.center) to (112);
		\draw [in=90, out=-90, looseness=1.25] (120.center) to (113);
		\draw [in=-90, out=90, looseness=1.25] (114.center) to (121.center);
		\draw [in=-90, out=90, looseness=1.25] (115.center) to (122.center);
		\draw (116.center) to (123.center);
		\draw [in=90, out=-180, looseness=0.75] (133.center) to (131.center);
		\draw [in=90, out=0, looseness=0.75] (133.center) to (132.center);
		\draw (131.center) to (132.center);
		\draw [in=90, out=-150, looseness=0.75] (141.center) to (139.center);
		\draw [in=90, out=-30, looseness=0.75] (141.center) to (140.center);
		\draw [bend left=60, looseness=1.25] (139.center) to (140.center);
		\draw [in=90, out=-150, looseness=0.75] (149.center) to (147.center);
		\draw [in=90, out=-30, looseness=0.75] (149.center) to (148.center);
		\draw [bend left=60, looseness=1.25] (147.center) to (148.center);
		\draw [in=-90, out=90] (152) to (158.center);
		\draw [in=90, out=-90] (160.center) to (153);
		\draw [in=90, out=-90, looseness=1.25] (159.center) to (154);
		\draw [in=90, out=-90, looseness=1.25] (161.center) to (154);
		\draw [in=-90, out=90, looseness=1.25] (155.center) to (162.center);
		\draw [in=-90, out=90, looseness=1.25] (156.center) to (163.center);
		\draw (157.center) to (164.center);
		\draw [in=90, out=-180, looseness=0.75] (170.center) to (168.center);
		\draw [in=90, out=0, looseness=0.75] (170.center) to (169.center);
		\draw (168.center) to (169.center);
		\draw [in=90, out=-150, looseness=0.75] (178.center) to (176.center);
		\draw [in=90, out=-30, looseness=0.75] (178.center) to (177.center);
		\draw [bend left=60, looseness=1.25] (176.center) to (177.center);
		\draw [in=90, out=-150, looseness=0.75] (186.center) to (184.center);
		\draw [in=90, out=-30, looseness=0.75] (186.center) to (185.center);
		\draw [bend left=60, looseness=1.25] (184.center) to (185.center);
		\draw [in=-90, out=90] (189) to (190.center);
	\end{pgfonlayer}
\end{tikzpicture}
\end{equation}
On the left-hand-side tree, 
the minimum cost of a successful attack is $\$100$, 
whereas on the right-hand-side tree, it is $\$110$.
However, if we were to view these two attack trees as ordinary syntax graphs, 
where one freely merges isomorphic subgraphs, 
they would both represent the same term:
    $$ \AND_2(\OR_2(\mathtt{D}, \mathtt{F}), \OR_2(\mathtt{F}, \mathtt{S})). $$
The ability to explicitly treat sharing is the first key difference between ordinary, algebraic terms and term graphs.

The second key difference is that term graphs, in the sense of \Cref{def:term-graph} below, 
allow for multiple \emph{inputs} and \emph{outputs}.
This enables us to compose term graphs such as attack trees from smaller ones,
generalising the \emph{modular composition} of attack trees \cite[Definition II.4]{lopuhaa2024attack}.
For example: 

\begin{equation}\label{eq:composition-example}
    \begin{tikzpicture}
	\begin{pgfonlayer}{nodelayer}
		\node [style=none] (27) at (1, 0) {$=$};
		\node [style=none] (45) at (-4, 0) {$\circ$};
		\node [style=new style 0] (70) at (-3, -1.25) {};
		\node [style=new style 0] (71) at (-1, -1.25) {};
		\node [style=none] (72) at (-2, 1) {};
		\node [style=none] (73) at (-2.25, 0.25) {};
		\node [style=none] (74) at (-1.75, 0.25) {};
		\node [style=none] (75) at (-2, 1.75) {};
		\node [style=none] (84) at (-0.5, 0) {};
		\node [style=none] (85) at (-0.5, 1.75) {};
		\node [style={blank_rectangle}] (86) at (-6.75, -1.25) {};
		\node [style=new style 0] (87) at (-5.25, -1.25) {};
		\node [style=none] (88) at (-6, 1) {};
		\node [style=none] (89) at (-6.25, 0.25) {};
		\node [style=none] (90) at (-5.75, 0.25) {};
		\node [style=none] (91) at (-6.75, 2) {};
		\node [style=none] (100) at (-7, 3) {};
		\node [style=none] (102) at (-7, 3.5) {};
		\node [style=none] (103) at (-7.25, 2) {};
		\node [style={blank_rectangle}] (110) at (-7.75, -1.25) {};
		\node [style=new style 0] (111) at (2.25, -1.25) {};
		\node [style=new style 0] (112) at (6.25, -1.25) {};
		\node [style=new style 0] (113) at (4.25, -1.25) {};
		\node [style=none] (114) at (3.25, 1) {};
		\node [style=none] (115) at (5.25, 1) {};
		\node [style=none] (116) at (4.25, 3) {};
		\node [style=none] (117) at (3, 0.25) {};
		\node [style=none] (118) at (5, 0.25) {};
		\node [style=none] (119) at (5.5, 0.25) {};
		\node [style=none] (120) at (3.5, 0.25) {};
		\node [style=none] (121) at (4, 2) {};
		\node [style=none] (122) at (4.5, 2) {};
		\node [style=none] (123) at (4.25, 3.75) {};
		\node [style=none] (128) at (4.25, 3) {};
		\node [style=none] (129) at (4, 2) {};
		\node [style=none] (130) at (4.5, 2) {};
		\node [style=none] (131) at (3.75, 2) {};
		\node [style=none] (132) at (4.75, 2) {};
		\node [style=none] (133) at (4.25, 3) {};
		\node [style=none] (134) at (4, 2) {};
		\node [style=none] (135) at (4.5, 2) {};
		\node [style=none] (136) at (5.25, 1) {};
		\node [style=none] (137) at (5, 0.25) {};
		\node [style=none] (138) at (5.5, 0.25) {};
		\node [style=none] (139) at (4.75, 0) {};
		\node [style=none] (140) at (5.75, 0) {};
		\node [style=none] (141) at (5.25, 1) {};
		\node [style=none] (142) at (5, 0.25) {};
		\node [style=none] (143) at (5.5, 0.25) {};
		\node [style=none] (144) at (3.25, 1) {};
		\node [style=none] (145) at (3, 0.25) {};
		\node [style=none] (146) at (3.5, 0.25) {};
		\node [style=none] (147) at (2.75, 0) {};
		\node [style=none] (148) at (3.75, 0) {};
		\node [style=none] (149) at (3.25, 1) {};
		\node [style=none] (150) at (3, 0.25) {};
		\node [style=none] (151) at (3.5, 0.25) {};
		\node [style=none] (152) at (2.25, -2) {$\mathtt{D}$};
		\node [style=none] (153) at (4.25, -2) {$\mathtt{F}$};
		\node [style=none] (154) at (6.25, -2) {$\mathtt{S}$};
		\node [style=none] (155) at (-3, -2) {$\mathtt{D}$};
		\node [style=none] (156) at (-1, -2) {$\mathtt{F}$};
		\node [style=none] (157) at (-5.25, -2) {$\mathtt{S}$};
		\node [style=none] (158) at (-7, 3) {};
		\node [style=none] (159) at (-7.25, 2) {};
		\node [style=none] (160) at (-6.75, 2) {};
		\node [style=none] (161) at (-7, 3) {};
		\node [style=none] (162) at (-7.25, 2) {};
		\node [style=none] (163) at (-6.75, 2) {};
		\node [style=none] (164) at (-7.5, 2) {};
		\node [style=none] (165) at (-6.5, 2) {};
		\node [style=none] (166) at (-7, 3) {};
		\node [style=none] (167) at (-7.25, 2) {};
		\node [style=none] (168) at (-6.75, 2) {};
		\node [style=none] (169) at (-6, 1) {};
		\node [style=none] (170) at (-6.25, 0.25) {};
		\node [style=none] (171) at (-5.75, 0.25) {};
		\node [style=none] (172) at (-6, 1) {};
		\node [style=none] (173) at (-6.25, 0.25) {};
		\node [style=none] (174) at (-5.75, 0.25) {};
		\node [style=none] (175) at (-6.5, 0) {};
		\node [style=none] (176) at (-5.5, 0) {};
		\node [style=none] (177) at (-6, 1) {};
		\node [style=none] (178) at (-6.25, 0.25) {};
		\node [style=none] (179) at (-5.75, 0.25) {};
		\node [style=none] (180) at (-2, 1) {};
		\node [style=none] (181) at (-2.25, 0.25) {};
		\node [style=none] (182) at (-1.75, 0.25) {};
		\node [style=none] (183) at (-2, 1) {};
		\node [style=none] (184) at (-2.25, 0.25) {};
		\node [style=none] (185) at (-1.75, 0.25) {};
		\node [style=none] (186) at (-2, 1) {};
		\node [style=none] (187) at (-2.25, 0.25) {};
		\node [style=none] (188) at (-1.75, 0.25) {};
		\node [style=none] (189) at (-2.5, 0) {};
		\node [style=none] (190) at (-1.5, 0) {};
		\node [style=none] (191) at (-2, 1) {};
		\node [style=none] (192) at (-2.25, 0.25) {};
		\node [style=none] (193) at (-1.75, 0.25) {};
	\end{pgfonlayer}
	\begin{pgfonlayer}{edgelayer}
		\draw [in=-90, out=90, looseness=1.25] (70) to (73.center);
		\draw [in=-90, out=90, looseness=1.25] (71) to (74.center);
		\draw (72.center) to (75.center);
		\draw [in=-90, out=90, looseness=1.25] (71) to (84.center);
		\draw (85.center) to (84.center);
		\draw [in=-90, out=90, looseness=1.25] (86) to (89.center);
		\draw [in=-90, out=90, looseness=1.25] (87) to (90.center);
		\draw [in=-90, out=90, looseness=1.25] (88.center) to (91.center);
		\draw [in=-90, out=90, looseness=0.75] (100.center) to (102.center);
		\draw [in=-90, out=90, looseness=1.25] (110) to (103.center);
		\draw [in=-90, out=90] (111) to (117.center);
		\draw [in=90, out=-90] (119.center) to (112);
		\draw [in=90, out=-90, looseness=1.25] (118.center) to (113);
		\draw [in=90, out=-90, looseness=1.25] (120.center) to (113);
		\draw [in=-90, out=90, looseness=1.25] (114.center) to (121.center);
		\draw [in=-90, out=90, looseness=1.25] (115.center) to (122.center);
		\draw (116.center) to (123.center);
		\draw [in=90, out=-180, looseness=0.75] (133.center) to (131.center);
		\draw [in=90, out=0, looseness=0.75] (133.center) to (132.center);
		\draw (131.center) to (132.center);
		\draw [in=90, out=-150, looseness=0.75] (141.center) to (139.center);
		\draw [in=90, out=-30, looseness=0.75] (141.center) to (140.center);
		\draw [bend left=60, looseness=1.25] (139.center) to (140.center);
		\draw [in=90, out=-150, looseness=0.75] (149.center) to (147.center);
		\draw [in=90, out=-30, looseness=0.75] (149.center) to (148.center);
		\draw [bend left=60, looseness=1.25] (147.center) to (148.center);
		\draw [in=90, out=-180, looseness=0.75] (166.center) to (164.center);
		\draw [in=90, out=0, looseness=0.75] (166.center) to (165.center);
		\draw (164.center) to (165.center);
		\draw [in=90, out=-150, looseness=0.75] (177.center) to (175.center);
		\draw [in=90, out=-30, looseness=0.75] (177.center) to (176.center);
		\draw [bend left=60, looseness=1.25] (175.center) to (176.center);
		\draw [in=90, out=-150, looseness=0.75] (191.center) to (189.center);
		\draw [in=90, out=-30, looseness=0.75] (191.center) to (190.center);
		\draw [bend left=60, looseness=1.25] (189.center) to (190.center);
	\end{pgfonlayer}
\end{tikzpicture}
\end{equation}

Here, input nodes are represented by small squares (``$\square$'') 
and output nodes are the ends of wires not connected to any other node. 
Both are ordered and read left to right.

We now introduce term graphs and their compositional structure in detail,
following the general theory laid out in \cite{corradini1999algebraic} 
for everything not specific to attack trees. 
In the following definition, 
we use the notation $\listfun(X)$ for the set of lists (i.e.~finite, ordered sequences) 
over a set $X$, and $\toset{l}$ for the set of all elements appearing in a list $l$.

\begin{definition}[Term graph]\label{def:term-graph}
    A \emph{term graph} $T$ over a signature $\Sigma$, or \emph{$\Sigma$-term graph},
    is a tuple 
        $$ T=(N^T,\inputs^T,\outputs^T,\labelfun^T,\ch^T), $$ 
    where:
    \begin{enumerate}
        \item $N^T$ is a finite set of \emph{nodes}.
        \item $\inputs^T \in \listfun(N^T)$ is a list of \emph{input nodes}, and 
        \item $\outputs^T \in \listfun(N^T)$ is a list of \emph{output nodes}.
        \item $\labelfun^T \colon N^T \setminus \toset{\inputs^T} \to \Sigma $ is a function
        assigning to each non-input node a function symbol from $\Sigma$. 
        \item $\ch^T \colon N^T \setminus \toset{\inputs^T} \to \listfun(N)$ is a function assigning 
        to each non-input node $n$ its list of \emph{children} $\ch^T\!(n)$,
    \end{enumerate}
    These data are required to satisfy the following axioms.
    \begin{enumerate}
        \item For each non-input node $n\in N^T$, 
            $$\arity_\Sigma(\labelfun^T(n)) = \mathrm{length}(\ch^T(n)).$$
        \item The directed graph $(N^T,E^T)$ is acyclic, where 
            $$ E^T := \{ (n, m) \in N^T \times N^T \mid m \in \ch^T(n) \}.$$
        \item The list of input nodes $\inputs^T$ contains no duplicates.
    \end{enumerate}
    We write $T\colon i \to j$ for term graphs with $i$ inputs and $j$ outputs.     An \emph{isomorphism of term graphs} $S$ and $T$ is a bijection $N^S \to N^{T}$
    preserving the input and output nodes (including their ordering), 
    the labelling function, and the child function.
    We will generally consider term graphs up to isomorphism, treating isomorphic term graphs as equal.
    
    The set of all $\Sigma$-term graphs with $i$ inputs and $j$ outputs 
    (up to isomorphism) 
    is denoted $\termgraphs{\Sigma}(i, j)$.
\end{definition}

From this perspective, attack trees are term graphs with $0$ inputs and $1$ output: 

\begin{definition}[Attack tree]
    Let $B$ be a set (thought of as consisting of basic attack step labels).
    An \emph{attack tree} $T$ over $B$ is an $\mathsf{AT}(B)$-term graph $T\colon 0\to 1$.
\end{definition}

We may think of general $\mathsf{AT}(B)$-term graphs as representing \emph{components} of attack trees,
yet to be composed into a full attack tree.
We will therefore use the following terminology.

\begin{definition}[Attack tree component]\label{def:attack-tree-component}
    Let $B$ be a set.
    An \emph{attack tree component} over $B$ is an $\mathsf{AT}(B)$-term graph.
\end{definition}

\begin{example}
    Consider the following attack tree component over $B= \{ \mathtt{D}, \mathtt{F}, \mathtt{S} \}$:
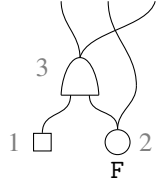
    \begin{equation*}
        \begin{tikzpicture}
	\begin{pgfonlayer}{nodelayer}
		\node [style={blank_rectangle}] (70) at (-1.5, -1.25) {};
		\node [style=new style 0] (71) at (0.5, -1.25) {};
		\node [style=none] (72) at (-0.5, 1) {};
		\node [style=none] (73) at (-0.75, 0) {};
		\node [style=none] (74) at (-0.25, 0) {};
		\node [style=none] (75) at (-1, 2.5) {};
		\node [style=none] (76) at (-0.5, 1) {};
		\node [style=none] (77) at (-0.75, 0) {};
		\node [style=none] (78) at (-0.25, 0) {};
		\node [style=none] (79) at (-1, 0) {};
		\node [style=none] (80) at (0, 0) {};
		\node [style=none] (81) at (-0.5, 1) {};
		\node [style=none] (82) at (-0.75, 0) {};
		\node [style=none] (83) at (-0.25, 0) {};
		\node [style=none] (84) at (1, 0.25) {};
		\node [style=none] (85) at (0.25, 2.5) {};
		\node [style=none] (156) at (0.5, -2) {$\mathtt{F}$};
		\node [style=none] (157) at (0, 2.75) {};
		\node [style=none] (158) at (-0.25, -2.25) {};
		\node [style=none] (159) at (-1.5, 0.75) {\textcolor{gray}{3}};
		\node [style=none] (160) at (-2.25, -1.25) {\textcolor{gray}{1}};
		\node [style=none] (161) at (1.25, -1.25) {\textcolor{gray}{2}};
		\node [style=none] (162) at (1.5, 2.5) {};
	\end{pgfonlayer}
	\begin{pgfonlayer}{edgelayer}
		\draw [in=-90, out=90, looseness=1.25] (70) to (73.center);
		\draw [in=-90, out=90, looseness=1.25] (71) to (74.center);
		\draw [in=-90, out=90] (72.center) to (75.center);
		\draw [in=90, out=-180, looseness=0.75] (81.center) to (79.center);
		\draw [in=90, out=0, looseness=0.75] (81.center) to (80.center);
		\draw (79.center) to (80.center);
		\draw [in=-90, out=90] (71) to (84.center);
		\draw [in=90, out=-90] (85.center) to (84.center);
		\draw [in=-105, out=75, looseness=0.75] (81.center) to (162.center);
	\end{pgfonlayer}
\end{tikzpicture}
    \end{equation*}
    It can be defined more precisely as a term graph $T\colon 1 \to 3$ with one input and three outputs, 
    and the following data: 
    \begin{enumerate}
        \item Three nodes, arbitrarily chosen to be: $N^T := \{1, 2, 3\}$,
        \item A single-element list of input nodes: $\inputs^T := [1]$,
        \item A three-element list of output nodes: $\outputs^T := [3, 2, 3]$,
        \item A labelling function $\labelfun^T\colon N^T \setminus \toset{\inputs^T} \to \mathsf{AT}(B)$ given by $\labelfun^T(2) = \mathtt{F},\labelfun^T(3) = \AND_2$.
        \item A child function $\ch^T\colon N^T \setminus \toset{\inputs^T} \to \listfun(N^T)$ given by $\ch^T(2) = [\,], \ch^T(3) = [1, 2]$.
    \end{enumerate}
    Note that the graphical representation of $T$ is read bottom to top, 
    with squares representing input nodes. 
    Output nodes, on the other hand, are depicted by all those nodes that have outgoing wires towards the top of the diagram, 
    not connected to any other node.
    In particular, the list of output nodes may contain duplicates, 
    whereas the list of input nodes may not (the third of the axioms in \Cref{def:term-graph}).
\end{example}

\subsection{The sequential composition of term graphs}

The reason we consider term graphs with multiple inputs and outputs, 
and attack tree components in particular, is 
that it allows us to understand the compositional structure of attack trees.
There are two ways to compose term graphs: 
\emph{sequential composition} and \emph{parallel composition}.

We start with the former, 
an example for which was already given in \eqref{eq:composition-example}.
The sequential composition is given by gluing the output nodes of the first term graph
to the input nodes of the second term graph. More formally:

\begin{definition}[Sequential composition of term graphs]\label{def:sequential-composition-of-term-graphs}
    Let $S\colon i\to j,T\colon j \to k$ be term graphs over a signature $\Sigma$. 
    Their \emph{(sequential) composition} is 
        $ T \circ S\colon i \to k, $
    where: 
    \begin{enumerate}
        \item $N^{T \circ S} := \left(N^{S} \sqcup N^{T}\right) / \sim $.
        Here, $\sim$ is the equivalence relation that identifies the $l$-th output node
        of $S$ with the $l$-th input node of $T$, for all $l\in \{1,\dots, j\}$.
        \item $\inputs^{T \circ S}:= \phi^S_*[\inputs^{S}]$ and $\outputs^{T \circ S}:= \phi^T_*[\outputs^T]$,
        where $\phi^S\colon N^{S} \to N^{T \circ S}$, $\phi^T: N^{T} \to N^{T \circ S}$ 
        are the inclusion maps followed by the canonical projection onto the quotient, 
        and $f_*\colon \listfun(X) \to \listfun(Y) $ denotes the extension of a map $f:X\to Y$ between sets $X,Y$ to lists.
        \item $\labelfun^{T \circ S} \colon N^{T \circ S} \setminus \toset{\inputs^{T \circ S}} \to \Sigma$
        is given by
            $$ 
            \labelfun^{T \circ S}([n]) := \begin{cases}
                \labelfun^{T}(n) \text{ if } n \in N^T \setminus \toset{\inputs^T} \\
                \labelfun^{S}(n) \text{ if } n \in N^S \setminus \toset{\inputs^S},
            \end{cases}
            $$
        for all $[n] \in N^{T \circ S}$
        \item Similarly, $\ch^{T \circ S} \colon N^{T \circ S} \to \listfun(N^{T \circ S})$
        is given by
            $$ 
            \ch^{T \circ S}([n]) := \begin{cases}
                \ch^{T}(n) \text{ if } n \in N^T \setminus \toset{\inputs^T} \\
                \ch^{S}(n) \text{ if } n \in N^S \setminus \toset{\inputs^S},
            \end{cases}
            $$
        for all $[n] \in N^{T \circ S}$
    \end{enumerate} 
    Note that both $\labelfun^{T \circ S}$ and $\ch^{T \circ S}$ 
    are well-defined this way, by the definitions of $N^{T \circ S}$ and $\inputs^{T \circ S}$.
\end{definition}

\subsection{The parallel composition of term graphs}

The parallel composition is given by the disjoint union, 
concatenating the respective lists of input and output nodes:
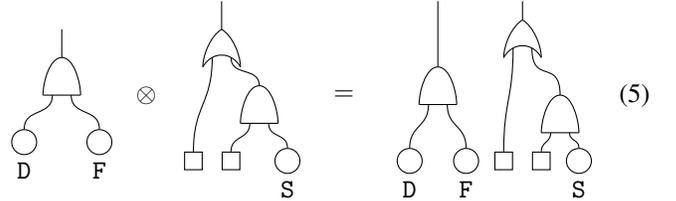
\begin{equation}\label{eq:parallel-composition-example}
    \begin{tikzpicture}
	\begin{pgfonlayer}{nodelayer}
		\node [style=none] (27) at (1, 0) {$=$};
		\node [style=none] (45) at (-4.25, 0) {$\otimes$};
		\node [style=new style 0] (72) at (-7.5, -1.25) {};
		\node [style=new style 0] (73) at (-5.5, -1.25) {};
		\node [style=none] (74) at (-6.5, 1) {};
		\node [style=none] (75) at (-6.75, 0) {};
		\node [style=none] (76) at (-6.25, 0) {};
		\node [style=none] (77) at (-6.5, 1.75) {};
		\node [style=none] (78) at (-6.5, 1) {};
		\node [style=none] (79) at (-6.75, 0) {};
		\node [style=none] (80) at (-6.25, 0) {};
		\node [style=none] (81) at (-7, 0) {};
		\node [style=none] (82) at (-6, 0) {};
		\node [style=none] (83) at (-6.5, 1) {};
		\node [style=none] (84) at (-6.75, 0) {};
		\node [style=none] (85) at (-6.25, 0) {};
		\node [style={blank_rectangle}] (88) at (-2, -1.75) {};
		\node [style=new style 0] (89) at (-0.5, -1.75) {};
		\node [style=none] (90) at (-1.25, 0.25) {};
		\node [style=none] (91) at (-1.5, -0.75) {};
		\node [style=none] (92) at (-1, -0.75) {};
		\node [style=none] (93) at (-2, 1) {};
		\node [style=none] (94) at (-1.25, 0.25) {};
		\node [style=none] (95) at (-1.5, -0.75) {};
		\node [style=none] (96) at (-1, -0.75) {};
		\node [style=none] (97) at (-1.75, -0.75) {};
		\node [style=none] (98) at (-0.75, -0.75) {};
		\node [style=none] (99) at (-1.25, 0.25) {};
		\node [style=none] (100) at (-1.5, -0.75) {};
		\node [style=none] (101) at (-1, -0.75) {};
		\node [style=none] (102) at (-2.25, 1.75) {};
		\node [style=none] (103) at (-2, 1) {};
		\node [style=none] (104) at (-2.25, 2.5) {};
		\node [style=none] (105) at (-2.5, 1) {};
		\node [style=none] (106) at (-2.25, 1.75) {};
		\node [style=none] (107) at (-2, 1) {};
		\node [style=none] (108) at (-2.75, 0.75) {};
		\node [style=none] (109) at (-1.75, 0.75) {};
		\node [style=none] (110) at (-2.25, 1.75) {};
		\node [style=none] (111) at (-2, 1) {};
		\node [style={blank_rectangle}] (112) at (-3, -1.75) {};
		\node [style=none] (153) at (-7.5, -2) {$\mathtt{D}$};
		\node [style=none] (154) at (-5.5, -2) {$\mathtt{F}$};
		\node [style=none] (155) at (-0.5, -2.5) {$\mathtt{S}$};
		\node [style=new style 0] (157) at (2.75, -1.75) {};
		\node [style=new style 0] (158) at (4.25, -1.75) {};
		\node [style=none] (159) at (3.5, 0.75) {};
		\node [style=none] (160) at (3.25, -0.25) {};
		\node [style=none] (161) at (3.75, -0.25) {};
		\node [style=none] (162) at (3.5, 2.5) {};
		\node [style=none] (163) at (3.5, 0.75) {};
		\node [style=none] (164) at (3.25, -0.25) {};
		\node [style=none] (165) at (3.75, -0.25) {};
		\node [style=none] (166) at (3, -0.25) {};
		\node [style=none] (167) at (4, -0.25) {};
		\node [style=none] (168) at (3.5, 0.75) {};
		\node [style=none] (169) at (3.25, -0.25) {};
		\node [style=none] (170) at (3.75, -0.25) {};
		\node [style={blank_rectangle}] (171) at (6.25, -1.75) {};
		\node [style=new style 0] (172) at (7.25, -1.75) {};
		\node [style=none] (173) at (6.75, 0) {};
		\node [style=none] (174) at (6.5, -1) {};
		\node [style=none] (175) at (7, -1) {};
		\node [style=none] (176) at (6, 1.25) {};
		\node [style=none] (177) at (6.75, 0) {};
		\node [style=none] (178) at (6.5, -1) {};
		\node [style=none] (179) at (7, -1) {};
		\node [style=none] (180) at (6.25, -1) {};
		\node [style=none] (181) at (7.25, -1) {};
		\node [style=none] (182) at (6.75, 0) {};
		\node [style=none] (183) at (6.5, -1) {};
		\node [style=none] (184) at (7, -1) {};
		\node [style=none] (185) at (5.75, 2) {};
		\node [style=none] (186) at (6, 1.25) {};
		\node [style=none] (187) at (5.75, 2.5) {};
		\node [style=none] (188) at (5.5, 1.25) {};
		\node [style=none] (189) at (5.75, 2) {};
		\node [style=none] (190) at (6, 1.25) {};
		\node [style=none] (191) at (5.25, 1) {};
		\node [style=none] (192) at (6.25, 1) {};
		\node [style=none] (193) at (5.75, 2) {};
		\node [style=none] (194) at (6, 1.25) {};
		\node [style={blank_rectangle}] (195) at (5.25, -1.75) {};
		\node [style=none] (196) at (2.75, -2.5) {$\mathtt{D}$};
		\node [style=none] (197) at (4.25, -2.5) {$\mathtt{F}$};
		\node [style=none] (198) at (7.25, -2.5) {$\mathtt{S}$};
	\end{pgfonlayer}
	\begin{pgfonlayer}{edgelayer}
		\draw [in=-90, out=90, looseness=1.25] (72) to (75.center);
		\draw [in=-90, out=90, looseness=1.25] (73) to (76.center);
		\draw (74.center) to (77.center);
		\draw [in=90, out=-180, looseness=0.75] (83.center) to (81.center);
		\draw [in=90, out=0, looseness=0.75] (83.center) to (82.center);
		\draw (81.center) to (82.center);
		\draw [in=-90, out=90, looseness=1.25] (88) to (91.center);
		\draw [in=-90, out=90, looseness=1.25] (89) to (92.center);
		\draw [in=-90, out=90, looseness=1.25] (90.center) to (93.center);
		\draw [in=90, out=-180, looseness=0.75] (99.center) to (97.center);
		\draw [in=90, out=0, looseness=0.75] (99.center) to (98.center);
		\draw (97.center) to (98.center);
		\draw [in=-90, out=90, looseness=0.75] (102.center) to (104.center);
		\draw [in=90, out=-150, looseness=0.75] (110.center) to (108.center);
		\draw [in=90, out=-30, looseness=0.75] (110.center) to (109.center);
		\draw [bend left=60, looseness=1.25] (108.center) to (109.center);
		\draw [in=-90, out=90, looseness=1.25] (112) to (105.center);
		\draw [in=-90, out=90, looseness=1.25] (157) to (160.center);
		\draw [in=-90, out=90, looseness=1.25] (158) to (161.center);
		\draw (159.center) to (162.center);
		\draw [in=90, out=-180, looseness=0.75] (168.center) to (166.center);
		\draw [in=90, out=0, looseness=0.75] (168.center) to (167.center);
		\draw (166.center) to (167.center);
		\draw [in=-90, out=90, looseness=1.25] (171) to (174.center);
		\draw [in=-90, out=90, looseness=1.25] (172) to (175.center);
		\draw [in=-90, out=90, looseness=1.75] (173.center) to (176.center);
		\draw [in=90, out=-180, looseness=0.75] (182.center) to (180.center);
		\draw [in=90, out=0, looseness=0.75] (182.center) to (181.center);
		\draw (180.center) to (181.center);
		\draw [in=-90, out=90, looseness=0.75] (185.center) to (187.center);
		\draw [in=90, out=-150, looseness=0.75] (193.center) to (191.center);
		\draw [in=90, out=-30, looseness=0.75] (193.center) to (192.center);
		\draw [bend left=60, looseness=1.25] (191.center) to (192.center);
		\draw [in=-90, out=90, looseness=1.25] (195) to (188.center);
	\end{pgfonlayer}
\end{tikzpicture}
\end{equation}
The full definition is as follows.

\begin{definition}[Parallel composition of term graphs]\label{def:parallel-composition-of-term-graphs}
    Let $S\colon i\to j,T\colon k \to l$ be term graphs over a signature $\Sigma$. 
    Their \emph{parallel composition} is 
        $ T \otimes S\colon i + k \to j + l, $
    where: 
    \begin{enumerate}
        \item $N^{T \otimes S} := N^{S} \sqcup N^{T}$.
        \item $\inputs^{T \otimes S}:= \inputs^{S} \cdot \inputs^{T}$ and 
        $\outputs^{T \otimes S}:= \outputs^{S} \cdot \outputs^{T}$.
        \item $\labelfun^{T \otimes S}\colon N^{T \otimes S} \setminus \toset{\inputs^{T \otimes S}} \to \Sigma$
        is given by
            $$ 
            \labelfun^{T \otimes S}(n) := \begin{cases}
                \labelfun^{T}(n) \text{ if } n \in N^T \setminus \toset{\inputs^T} \\
                \labelfun^{S}(n) \text{ if } n \in N^S \setminus \toset{\inputs^S},
            \end{cases}
            $$
        for all $n \in N^{T \otimes S}$.
        \item Similarly, $\ch^{T \otimes S}\colon N^{T \otimes S} \to \listfun(N^{T \otimes S})$
        is given by
            $$ 
            \ch^{T \otimes S}(n) := \begin{cases}
                \ch^{T}(n) \text{ if } n \in N^T \setminus \toset{\inputs^T} \\
                \ch^{S}(n) \text{ if } n \in N^S \setminus \toset{\inputs^S},
            \end{cases} 
            $$
        for all $n \in N^{T \otimes S}$.
    \end{enumerate}
\end{definition}

\subsection{Atomic term graphs}

We will show that every term graph can be constructed, via parallel and sequential composition, from so-called \emph{atomic} term graphs. These consist of term graphs of function symbols, plus a few structural ones which we describe below.


\subsubsection{Identity term graphs}

The simplest term graphs
 are the \emph{identity term graphs}, 
$\id_0\colon 0 \to 0$ and $\id_1\colon 1 \to 1$. 
Here, $\id_0$ is the empty term graph, and 
 $\id_1$ has exactly one node, which is both input and output:
\begin{equation*}
    \begin{tikzpicture}
	\begin{pgfonlayer}{nodelayer}
		\node [style=none] (53) at (1.25, 0.75) {};
		\node [style=none] (56) at (0, 0) {$:=$};
		\node [style=none] (57) at (-1.25, 0) {$\id_1$};
		\node [style={blank_rectangle}] (58) at (1.25, -0.5) {};
		\node [style=none] (59) at (0, -0.75) {};
	\end{pgfonlayer}
	\begin{pgfonlayer}{edgelayer}
		\draw (58) to (53.center);
	\end{pgfonlayer}
\end{tikzpicture}
\end{equation*}
Note that $\id_0$ is the identity for parallel composition, and the parallel composition of an appropriate amount of copies of $\id_1$ is the identity for sequential composition.
\subsubsection{Term graphs of function symbols}

Each function symbol $f$ in the signature $\Sigma$
has an \emph{associated term graph} $\langle f \rangle$
with $\arity(f) + 1$ nodes as follows: 

\begin{equation*}
    \begin{tikzpicture}
	\begin{pgfonlayer}{nodelayer}
		\node [style={blank_rectangle}] (52) at (2.25, 0.25) {$\;\;f\;\;$};
		\node [style=none] (53) at (2.25, 2) {};
		\node [style=none] (56) at (-1.25, 0) {$:=$};
		\node [style=none] (57) at (-2.75, 0) {$\langle f \rangle$};
		\node [style={blank_rectangle}] (58) at (-0.25, -1.5) {};
		\node [style={blank_rectangle}] (59) at (0.75, -1.5) {};
		\node [style={blank_rectangle}] (60) at (3.75, -1.5) {};
		\node [style={blank_rectangle}] (61) at (4.75, -1.5) {};
		\node [style=none] (62) at (2.25, -1.5) {$\dots$};
		\node [style=none] (63) at (2.25, -1) {};
		\node [style=none] (64) at (0, -2) {};
	\end{pgfonlayer}
	\begin{pgfonlayer}{edgelayer}
		\draw (53.center) to (52);
		\draw [in=-135, out=60] (58) to (52);
		\draw [in=-120, out=75, looseness=0.75] (59) to (52);
		\draw [in=-60, out=105] (60) to (52);
		\draw [in=-45, out=120, looseness=0.75] (61) to (52);
		\draw (63.center) to (52);
	\end{pgfonlayer}
\end{tikzpicture}
\end{equation*}

\subsubsection{Copy, delete, swap}\label{sec:copy-and-delete-term-graphs}

The term graphs $\copymap\colon 1\to 2$, $\delmap\colon 1\to 0$, $\swap\colon 2 \to 2$ are
the following term graphs:
\mlz{Can you make this picture prettier?}
\begin{equation*}
    \begin{tikzpicture}
	\begin{pgfonlayer}{nodelayer}
		\node [style=none] (53) at (-2.75, 1.25) {};
		\node [style=none] (56) at (-5.5, 0) {$:=$};
		\node [style=none] (57) at (-7, 0) {$\copymap$};
		\node [style=none] (65) at (-4.75, 1.25) {};
		\node [style=none] (67) at (0, 0) {$:=$};
		\node [style=none] (68) at (-1.25, 0) {$\delmap$};
		\node [style={blank_rectangle}] (70) at (-3.75, -1) {};
		\node [style={blank_rectangle}] (71) at (1, 0) {};
		\node [style=none] (72) at (0, -1.25) {};
		\node [style=none] (73) at (6.75, 1.25) {};
		\node [style=none] (74) at (4.5, 0) {$:=$};
		\node [style=none] (75) at (3, 0) {$\swap$};
		\node [style=none] (76) at (5.5, 1.25) {};
		\node [style={blank_rectangle}] (77) at (5.5, -1) {};
		\node [style={blank_rectangle}] (78) at (6.75, -1) {};
		\node [style=none] (79) at (0, 1.5) {};
	\end{pgfonlayer}
	\begin{pgfonlayer}{edgelayer}
		\draw [in=-90, out=90] (70) to (65.center);
		\draw [in=-90, out=90] (70) to (53.center);
		\draw [in=-90, out=90] (77) to (73.center);
		\draw [in=90, out=-90] (76.center) to (78);
	\end{pgfonlayer}
\end{tikzpicture}
\end{equation*}

In words, $\copymap$ is a term graph with one input that is also two outputs, $\delmap$ has one input and no outputs, and $\swap$ has two inputs which are also outputs, but interchanged. Note that despite their names and appearance, these are not (yet) maps that duplicate inputs, delete inputs, and swap inputs; that will come later, when we interpret these terms graphs as operators in a channel category in Section \ref{sec:channels}.






\begin{definition}
    An \emph{atomic term graph} over the signature $\Sigma$
    is a term graph $E$ such that  
        $$ E \in \{\id_0, \id_1,\copymap, \delmap, \swap\} \cup \{ \langle f \rangle \mid f\in \Sigma\}. $$
\end{definition}

\subsection{Decomposition into atomic term graphs}
The following theorem is \cite[Theorem 9]{corradini1999algebraic}.

\begin{theorem}\label{thm:decomposition-into-atomic-term-graphs}
    Every term graph is a sequential composite of parallel composites of 
    atomic term graphs. 

    In more detail, let $T$ be a term graph over a signature $\Sigma$. 
    Then there exist $N$, $k_1, \dots, k_N\in \mathbb{Z}_{\geq 0}$,
    and atomic term graphs 
    $A_{1,1}, \dots, A_{k_1,1}, \dots A_{1,N}, \dots A_{k_N,N}$ over $\Sigma$
    such that 
        $$ T = \mathop{\bigcirc}\limits_{i = 1}^{N} \bigotimes_{j = 1}^{k_i} A_{i,j}. $$ 
\end{theorem}

\begin{example}\label{ex:decomposition-example}
    Consider the attack tree $T$ from the right hand side of \Cref{eq:composition-example}.
    It can be decomposed into atomic $\mathsf{AT}(B)$-term graphs as follows:
    \begin{equation}\label{eq:decomposition-example}
        \begin{tikzpicture}
	\begin{pgfonlayer}{nodelayer}
		\node [style=none] (27) at (0, 0) {$=$};
		\node [style=none] (107) at (4, 1.75) {$\otimes$};
		\node [style=none] (108) at (4, 2.75) {$\circ$};
		\node [style=new style 0] (109) at (1.5, -4.25) {};
		\node [style=new style 0] (110) at (4, -4.25) {};
		\node [style=new style 0] (111) at (6.5, -4.25) {};
		\node [style=none] (112) at (1.5, -3.5) {};
		\node [style=none] (113) at (4, -3.5) {};
		\node [style=none] (114) at (6.5, -3.5) {};
		\node [style=none] (115) at (2.75, -4.25) {$\otimes$};
		\node [style=none] (116) at (5.25, -4.25) {$\otimes$};
		\node [style=none] (117) at (4, -2.75) {$\circ$};
		\node [style={blank_rectangle}] (118) at (1.75, -2) {};
		\node [style={blank_rectangle}] (119) at (4, -2) {};
		\node [style={blank_rectangle}] (120) at (6.25, -2) {};
		\node [style=none] (121) at (1.75, -1) {};
		\node [style=none] (122) at (3.5, -1) {};
		\node [style=none] (123) at (6.25, -1) {};
		\node [style=none] (124) at (2.75, -1.75) {$\otimes$};
		\node [style=none] (125) at (5.25, -1.75) {$\otimes$};
		\node [style=none] (126) at (4, -0.5) {$\circ$};
		\node [style=none] (127) at (4.5, -1) {};
		\node [style=none] (128) at (-6.5, 0) {$=$};
		\node [style=none] (129) at (-7.75, 0) {$T$};
		\node [style=new style 0] (130) at (-5.5, -2) {};
		\node [style=new style 0] (131) at (-1.5, -2) {};
		\node [style=new style 0] (132) at (-3.5, -2) {};
		\node [style=none] (133) at (-4.5, 0.25) {};
		\node [style=none] (134) at (-2.5, 0.25) {};
		\node [style=none] (135) at (-3.5, 2.25) {};
		\node [style=none] (136) at (-4.75, -0.5) {};
		\node [style=none] (137) at (-2.75, -0.5) {};
		\node [style=none] (138) at (-2.25, -0.5) {};
		\node [style=none] (139) at (-4.25, -0.5) {};
		\node [style=none] (140) at (-3.75, 1.25) {};
		\node [style=none] (141) at (-3.25, 1.25) {};
		\node [style=none] (142) at (-3.5, 3) {};
		\node [style=none] (143) at (-3.5, 2.25) {};
		\node [style=none] (144) at (-3.75, 1.25) {};
		\node [style=none] (145) at (-3.25, 1.25) {};
		\node [style=none] (146) at (-4, 1.25) {};
		\node [style=none] (147) at (-3, 1.25) {};
		\node [style=none] (148) at (-3.5, 2.25) {};
		\node [style=none] (149) at (-3.75, 1.25) {};
		\node [style=none] (150) at (-3.25, 1.25) {};
		\node [style=none] (151) at (-2.5, 0.25) {};
		\node [style=none] (152) at (-2.75, -0.5) {};
		\node [style=none] (153) at (-2.25, -0.5) {};
		\node [style=none] (154) at (-3, -0.75) {};
		\node [style=none] (155) at (-2, -0.75) {};
		\node [style=none] (156) at (-2.5, 0.25) {};
		\node [style=none] (157) at (-2.75, -0.5) {};
		\node [style=none] (158) at (-2.25, -0.5) {};
		\node [style=none] (159) at (-4.5, 0.25) {};
		\node [style=none] (160) at (-4.75, -0.5) {};
		\node [style=none] (161) at (-4.25, -0.5) {};
		\node [style=none] (162) at (-5, -0.75) {};
		\node [style=none] (163) at (-4, -0.75) {};
		\node [style=none] (164) at (-4.5, 0.25) {};
		\node [style=none] (165) at (-4.75, -0.5) {};
		\node [style=none] (166) at (-4.25, -0.5) {};
		\node [style=none] (167) at (-5.5, -2.75) {$\mathtt{D}$};
		\node [style=none] (168) at (-3.5, -2.75) {$\mathtt{F}$};
		\node [style=none] (169) at (-1.5, -2.75) {$\mathtt{S}$};
		\node [style={blank_rectangle}] (170) at (3.25, 3.5) {};
		\node [style={blank_rectangle}] (171) at (4.75, 3.5) {};
		\node [style=none] (172) at (4, 5.25) {};
		\node [style=none] (173) at (3.75, 4.25) {};
		\node [style=none] (174) at (4.25, 4.25) {};
		\node [style=none] (175) at (4, 5.75) {};
		\node [style=none] (176) at (4, 5.25) {};
		\node [style=none] (177) at (3.75, 4.25) {};
		\node [style=none] (178) at (4.25, 4.25) {};
		\node [style=none] (179) at (3.5, 4.25) {};
		\node [style=none] (180) at (4.5, 4.25) {};
		\node [style=none] (181) at (4, 5.25) {};
		\node [style=none] (182) at (3.75, 4.25) {};
		\node [style=none] (183) at (4.25, 4.25) {};
		\node [style={blank_rectangle}] (198) at (3, 0.25) {};
		\node [style={blank_rectangle}] (199) at (1.5, 0.25) {};
		\node [style=none] (200) at (2.25, 2) {};
		\node [style=none] (201) at (2, 1.25) {};
		\node [style=none] (202) at (2.5, 1.25) {};
		\node [style=none] (203) at (2.25, 2) {};
		\node [style=none] (204) at (2, 1.25) {};
		\node [style=none] (205) at (2.5, 1.25) {};
		\node [style=none] (206) at (1.75, 1) {};
		\node [style=none] (207) at (2.75, 1) {};
		\node [style=none] (208) at (2.25, 2) {};
		\node [style=none] (209) at (2, 1.25) {};
		\node [style=none] (210) at (2.5, 1.25) {};
		\node [style=none] (211) at (2.25, 2.5) {};
		\node [style=none] (212) at (1.5, -5) {$\mathtt{D}$};
		\node [style=none] (213) at (4, -5) {$\mathtt{F}$};
		\node [style=none] (214) at (6.5, -5) {$\mathtt{S}$};
		\node [style=none] (215) at (0, -5.75) {};
		\node [style=none] (216) at (6.5, -3) {};
		\node [style=none] (217) at (7.5, 0) {};
		\node [style={blank_rectangle}] (218) at (6.5, 0.25) {};
		\node [style={blank_rectangle}] (219) at (5, 0.25) {};
		\node [style=none] (220) at (5.75, 2) {};
		\node [style=none] (221) at (5.5, 1.25) {};
		\node [style=none] (222) at (6, 1.25) {};
		\node [style=none] (223) at (5.75, 2) {};
		\node [style=none] (224) at (5.5, 1.25) {};
		\node [style=none] (225) at (6, 1.25) {};
		\node [style=none] (226) at (5.25, 1) {};
		\node [style=none] (227) at (6.25, 1) {};
		\node [style=none] (228) at (5.75, 2) {};
		\node [style=none] (229) at (5.5, 1.25) {};
		\node [style=none] (230) at (6, 1.25) {};
		\node [style=none] (231) at (5.75, 2.5) {};
	\end{pgfonlayer}
	\begin{pgfonlayer}{edgelayer}
		\draw (109) to (112.center);
		\draw (110) to (113.center);
		\draw (111) to (114.center);
		\draw (118) to (121.center);
		\draw [in=-90, out=90, looseness=1.25] (119) to (122.center);
		\draw (120) to (123.center);
		\draw [in=-90, out=90, looseness=1.25] (119) to (127.center);
		\draw [in=-90, out=90] (130) to (136.center);
		\draw [in=90, out=-90] (138.center) to (131);
		\draw [in=90, out=-90, looseness=1.25] (137.center) to (132);
		\draw [in=90, out=-90, looseness=1.25] (139.center) to (132);
		\draw [in=-90, out=90, looseness=1.25] (133.center) to (140.center);
		\draw [in=-90, out=90, looseness=1.25] (134.center) to (141.center);
		\draw (135.center) to (142.center);
		\draw [in=90, out=-180, looseness=0.75] (148.center) to (146.center);
		\draw [in=90, out=0, looseness=0.75] (148.center) to (147.center);
		\draw (146.center) to (147.center);
		\draw [in=90, out=-150, looseness=0.75] (156.center) to (154.center);
		\draw [in=90, out=-30, looseness=0.75] (156.center) to (155.center);
		\draw [bend left=60, looseness=1.25] (154.center) to (155.center);
		\draw [in=90, out=-150, looseness=0.75] (164.center) to (162.center);
		\draw [in=90, out=-30, looseness=0.75] (164.center) to (163.center);
		\draw [bend left=60, looseness=1.25] (162.center) to (163.center);
		\draw [in=-90, out=90, looseness=1.25] (170) to (173.center);
		\draw [in=-90, out=90, looseness=1.25] (171) to (174.center);
		\draw (172.center) to (175.center);
		\draw [in=90, out=-180, looseness=0.75] (181.center) to (179.center);
		\draw [in=90, out=0, looseness=0.75] (181.center) to (180.center);
		\draw (179.center) to (180.center);
		\draw [in=90, out=-90, looseness=1.25] (202.center) to (198);
		\draw [in=90, out=-90, looseness=1.50] (201.center) to (199);
		\draw [in=90, out=-150, looseness=0.75] (208.center) to (206.center);
		\draw [in=90, out=-30, looseness=0.75] (208.center) to (207.center);
		\draw [bend left=60, looseness=1.25] (206.center) to (207.center);
		\draw (208.center) to (211.center);
		\draw [in=90, out=-90, looseness=1.25] (222.center) to (218);
		\draw [in=90, out=-90, looseness=1.50] (221.center) to (219);
		\draw [in=90, out=-150, looseness=0.75] (228.center) to (226.center);
		\draw [in=90, out=-30, looseness=0.75] (228.center) to (227.center);
		\draw [bend left=60, looseness=1.25] (226.center) to (227.center);
		\draw (228.center) to (231.center);
	\end{pgfonlayer}
\end{tikzpicture}
    \end{equation}
    We can also write this algebraically as,
    \begin{equation}\label{eq:decomposition-example-algebraic}
    \begin{aligned}
        T = \AND &\circ (\OR \otimes \OR) \\
                &\circ (\id_1 \otimes \copymap \otimes \id_1) \\
                &\circ (\mathtt{D} \otimes \mathtt{F} \otimes \mathtt{S}),
    \end{aligned}
    \end{equation}
    and, in fact, by \Cref{thm:decomposition-into-atomic-term-graphs}, 
    every attack tree can be written in a similar form.
    Note that this decomposition is not unique, 
    as we could equally write $T$ as,
    \begin{equation}\label{eq:shorter-decomposition}
        T = \AND \circ (\OR \otimes \OR) \circ (\mathtt{D} \otimes \copymap \otimes \mathtt{S}) \circ \mathtt{F}, 
    \end{equation}
    or even arbitrarily reorder parallel composites by introducing $\swap$ term graphs.
\end{example}

\section{Semantics: Channels}\label{sec:channels}

So far, we have only discussed syntax. 
Now, we investigate semantics:
what does a term graph \emph{denote}? 

To answer this question, 
we must first assign meaning to each function symbol 
in the given signature.
This is formalised by the notion of an \emph{interpretation}.
From such an interpretation, 
we directly obtain semantics 
for the atomic term graphs associated to each function symbol.
In order to lift these semantics to \emph{general}, composite term graphs,
we need to know the semantics of all further atomic term graphs,
and, most importantly, how to compose them.
This information is captured by the notion of a \emph{channel category}:
a structure consisting of sets of \emph{channels}, 
together with rules for composing them sequentially and in parallel.

Hence, the answer to the question of the semantics of term graphs is that,
in general, \emph{term graphs denote channels}. 
Examples of channels include ordinary functions, 
stochastic matrices (modelling conditional probability distributions, or, ``functions with random outputs''),
and term graphs themselves. 
What a ``channel'' is, in general, we describe \emph{axiomatically}, via the notion of a channel category.

\begin{remark}
    For readers familiar with the literature on \emph{gs-monoidal categories} 
    (see \cite{gadducci1996algebraic,corradini1999algebraic,fritz2023free}),
    we note that channel categories can be defined concisely as strict gs-monoidal categories generated by a single object.  
\end{remark}

The following definition does not require any prior knowledge of category theory.

\begin{definition}[Channel category]
    A \emph{channel category} $C$ consists of the following, for all $i,j,k,l \in \mathbb{Z}_{\geq 0}$:
    \begin{enumerate}
        \item A set $C(i, j)$, whose elements we call \emph{channels with $i$ inputs and $j$ outputs}.
        \item A function $\circ\colon C(i, j) \times C(j, k) \to C(i, k)$ 
        called \emph{(sequential) composition} of channels.
        \item A function $\otimes\colon C(i, j) \times C(k, l) \to C(i+k, j+l)$ called \emph{parallel composition} of channels.
        \item A number of special channels:
        \begin{enumerate}
            \item $\mathsf{id}_i \in C(i, i)$, the \emph{identity channels},
            \item $\mathsf{swap}_{i,j} \in C(i+j, i+j)$, the \emph{swap gates}.
            \item $\mathsf{copy} \in C(1, 2)$, the \emph{copy gate},
            \item $\mathsf{del} \in C(1, 0)$, the \emph{delete} or \emph{discard gate},
        \end{enumerate}
    \end{enumerate}
    These data are required to satisfy the following axioms, for all $i,j,k,l,m,n\in\mathbb{Z}_{\geq 0}$: 
    \begin{enumerate}
        \item (Associativity and unitality of sequential composition) 
        For all $f \in C(i, j)$, $g \in C(j, k)$, and $h \in C(k, l)$,  
            $$(h \circ g) \circ f = h \circ (g \circ f), \quad f \circ \mathsf{id}_i = f = \mathsf{id}_j \circ f.$$
        \item (Associativity and unitality of parallel composition) 
        For all $f \in C(i_1, j_1)$, $g \in C(i_2, j_2)$, and $h \in C(i_3, j_3)$,  
            $$(h \otimes g) \otimes f = h \otimes (g \otimes f), \quad f \otimes \mathsf{id}_0 = f = \mathsf{id}_0 \otimes f.$$
        \item (Symmetry and functoriality of parallel composition)
        For all $f \in C(i, j)$ and $g \in C(k, l)$,
        \begin{align*}
        \mathsf{swap}_{l,j} \circ (g \otimes f) &= f \otimes g = (g \otimes f) \circ \mathsf{swap}_{i,k},\\
        \mathsf{swap}_{i,j} \circ \mathsf{swap}_{j, i} &= \mathsf{id}_{i+j},
        \end{align*}
        as well as, for all $f_1 \in C(i, j)$, $f_2 \in C(k, l)$, $g_1 \in C(j, m)$, and $g_2 \in C(l, n)$, 
            $$ (g_1 \otimes g_2) \circ (f_1 \otimes f_2) = (g_1 \circ f_1) \otimes (g_2 \circ f_2).$$
        \item (Commutative comonoid laws for copy and discard) 
        The copy and discard gates satisfy the following equations:
        \begin{align*}
        (\mathsf{copy} \otimes \mathsf{id}_1) \circ \mathsf{copy} 
        &= (\mathsf{id}_1 \otimes \mathsf{copy}) \circ \mathsf{copy},\\
        (\mathsf{del} \otimes \mathsf{id}_1) \circ \mathsf{copy} 
        &= \mathsf{id}_1 
        =  (\mathsf{id}_1 \otimes \mathsf{del}) \circ \mathsf{copy},\\
        \mathsf{copy} &= \mathsf{swap}_{1,1} \circ \mathsf{copy}.
        \end{align*}
    \end{enumerate}
\end{definition}




\begin{example}[$\termgraphs{\Sigma}$]
    Let $\Sigma$ be a signature.
    The channel category $\termgraphs{\Sigma}$ of term graphs over $\Sigma$ is given as follows.
    For each $i,j\in \mathbb{Z}_{\geq 0}$, the set of channels $\termgraphs{\Sigma}(i,j)$ is 
    the set of term graphs with $i$ inputs and $j$ outputs.
    The sequential and parallel composition of term graphs are given as defined in \Cref{def:sequential-composition-of-term-graphs,def:parallel-composition-of-term-graphs}.
    The identity channel $\mathsf{id}_i$ is the parallel composition of $i$ copies of $\id_1$.
    The copy and delete gates are the atomic term graphs $\copymap$ and $\delmap$.
    Finally, the swap gate $\mathsf{swap}_{i,j}$ is given by the term graph that swaps the first $i$ inputs with the last $j$ ones.
\end{example}

\begin{example}[$\mathsf{Functions}_X$]
    Let $X$ be a set.
    Then the channel category $\mathsf{Functions}_X$ is given as follows. 
    For each $i,j\in \mathbb{Z}_{\geq 0}$, the set of channels $\mathsf{Functions}_X(i,j)$ is the set of functions $X^i \to X^j$.
    The sequential and parallel composition are given by the ordinary composition of functions and the cartesian product of functions, 
        $$ f \otimes g := f \times g, \;\; (f \times g)(x, y) :=  (f(x), g(y)) $$
    respectively, and the identity channel $\mathsf{id}_i$ is the identity function on $X^i$, for each $i\in \mathbb{Z}_{\geq 0}$.
    The copy gate is given by, 
        $$ \copymap: X \to X^2, \; \copymap(x) := (x, x), $$
    and the delete gate is the unique map $X\to \mathbf{1}$ from $X$ to the one-point set $\mathbf{1}=X^0$.
    Finally, the swap gate $\mathsf{swap}_{i,j}$ is given by the permutation $X^{i+j}\to X^{i+j}$ that swaps the first $i$ coordinates with the last $j$ ones. Note that $\mathsf{Functions}_X(0,1)$ is the set of functions $\mathbf{1} \rightarrow X$, which may be identified with $X$ itself.
\end{example}


\subsection{Semantics}
We now wish to define an \emph{attack tree semantics} 
as a mapping from attack tree components to some channel category that preserves the compositional structure of 
attack trees, as motivated in \Cref{sec:introduction,sec:example}.
What it means to ``preserve compositional structure'' is 
formalised by the concept of a \emph{functor of channel categories}.

\begin{definition}\label{def:functor_of_channel_categories}
    Let $C,D$ be channel categories. 
    A \emph{functor of channel categories} $F\colon C\to D$ is a family of functions $F\colon C(i, j) \to D(i, j)$ for each pair $i, j \in \mathbb{Z}_{\geq 0}$, 
    such that the following equations hold for all channels on which they are defined: 
    \begin{align*}
    F(g \circ f) &= F(g) \circ F(f), &
 F(\mathsf{id}_i) &= \mathsf{id}_i,\\
    F(g \otimes f) &= F(g) \otimes F(f), & F(\mathsf{swap}_{i,j}) &= \mathsf{swap}_{i,j},\\
    F(\mathsf{copy}) &= \mathsf{copy}, &
      F(\mathsf{del}) &= \mathsf{del},
    \end{align*}
\end{definition} 

With this definition established, we can now present the definition that constitutes our solution to \Cref{main-problem}.

\begin{definition}[Attack tree semantics and metrics] \label{def:semantics}
    Let $B$ be a set. A \emph{semantics for $\mathsf{AT}(B)$}, or a \emph{metric for $\mathsf{AT}(B)$}, is a pair $(C,F)$ of a channel category $C$ and a functor of channel categories $F\colon \termgraphs{\mathsf{AT}(B)} \rightarrow C$.
\end{definition}

The advantage of this definition is that it is general enough to capture various metrics and semantics from the literature (see Section \ref{sec:compositionality-of-attack-tree-metrics}). At the same time, the preservation of $\circ$ and $\otimes$ means that the decomposition of Theorem \ref{thm:decomposition-into-atomic-term-graphs} is preserved. When the target channel category is ``sufficiently quantitative'', such as the examples $\mathsf{Functions}$ and $\mathsf{Matrices}$ of Section \ref{sec:example}, this can be leveraged to obtain metric computation algorithms. We discuss this in detail in Theorem \ref{thm:algorithm}.

In our framework, \emph{there is no difference between AT metrics and semantics}. Semantics are thought of as being more qualitative while metrics are more quantitative, but this distinction is not a mathematical one.

Our category-theoretical formulation can easily be applied to extensions of the AT framework, such as \emph{dynamic ATs} \cite{jhawar2015attack}, attack-defence trees \cite{kordy2010foundations} and attack-fault trees \cite{kumar2017quantitative}, and fault tree extensions such as dynamic fault trees \cite{aslansefat2020dynamic}, as these are all directed acyclic graphs with a finite number of gate types. Interpreting existing metrics of such frameworks as channel categories, as we do for standard ATs in Section \ref{sec:compositionality-of-attack-tree-metrics}, is beyond the scope of this paper.

\subsection{Interpretations} 

The downside of Definition \ref{def:semantics} is that it does not explain how a functor of channel categories can be constructed. In this subsection, we show that such a functor is uniquely characterised by the image of the elements of its signature. The map from a signature to a channel category is called an \emph{interpretation}.


\begin{definition}[Interpretation of a signature]
    Let $\Sigma$ be a signature. 
    An \emph{interpretation} of $\Sigma$ in a channel category $C$
    is a map
        $$\mathcal{I}\colon \Sigma \to \bigsqcup_{i\in \mathbb{Z}_{\geq 0}} C(i, 1)$$
    such that for all $f\in \Sigma$, $\mathcal{I}(f) \in C(\arity(f), 1)$. Here, $\bigsqcup$ denotes the disjoint union of sets.
\end{definition}

\begin{example}\label{ex:boolean-functions-interpretation}
    Consider the signature $\mathsf{AT}(B)$ 
    of attack trees over a set of basic attack step labels $B$ 
    from \Cref{ex:attack-and-fault-tree-signatures}.
    Let $t\colon B \rightarrow \{0,1\}$ be an assignment of truth values to each basic attack step label, 
    which we interpret as indicating whether or not a basic attack step labelled with $b\in B$ is successful.
    For each such $t$, we may interpret 
    the basic attack step labels as their corresponding truth values, 
    and the logical connectives as their corresponding Boolean functions.
    This yields the following interpretation of $\mathsf{AT}(B)$ in the channel category $\mathsf{Functions}_{\{0,1\}}$:
    \begin{align*}
        \mathcal{I}_t\colon\mathsf{AT}(B) &\to \mathsf{Functions}_{\{0,1\}}\\
        \mathcal{I}_t(b)(\,) &:= t(b)   &&(b\in B)\\
        \mathcal{I}_t(\AND_i)(x_1, \dots, x_i) &:= \bigwedge_{k=1}^i x_k &&(i\in \mathbb{Z}_{\geq 0})\\
        \mathcal{I}_t(\OR_i)(x_1, \dots, x_i) &:= \bigvee_{k=1}^i x_k &&(i\in \mathbb{Z}_{\geq 0})
    \end{align*} 
    These equations show that our notion of an interpretation is well-aligned with the standard notion from logic.
\end{example}

The following theorem says that there is a \emph{unique} way to extend a given interpretation of any signature 
to a functor of channel categories; see \cite[Theorem 23]{corradini1999algebraic} for a proof in the setting of gs-monoidal categories. 

Because term graphs can be decomposed into atomic graphs and channel category functors preserve composition, functors are uniquely defined by the interpretation of the signature:

\begin{theorem}\label{thm:term-graphs-are-free-channel-category}\emph{\cite[Theorem 23]{corradini1999algebraic}}
    Let $\Sigma$ be a signature and let $C$ be a channel category,
    and let $\mathcal{I}$ be an interpretation of $\Sigma$ in $C$.
    Then there exists a unique functor of channel categories,
        $$\sem{\,\cdot\,}{\mathcal{I}}\colon \termgraphs{\Sigma} \to C,$$
    such that $\sem{f}{\mathcal{I}} = \mathcal{I}(f)$ for all $f \in \Sigma$. Conversely, for every functor of channel categories $F\colon \termgraphs{\Sigma} \rightarrow C$ there exists an interpretation $\mathcal{I}$ such that $F = \sem{\,\cdot\,}{\mathcal{I}}$.
\end{theorem}

This allows us to define the semantics of a term graph 
\emph{under} an interpretation.

\begin{definition}[Semantics under interpretation]
    Let $T$ be a term graph over a signature $\Sigma$,
    and let $\mathcal{I}$ be an interpretation of $\Sigma$ in a channel category $C$.
    The \emph{semantics} of $T$ under $\mathcal{I}$ is defined to be
    $\sem{T}{\mathcal{I}}$, i.e.~the value at $T$ of the unique functor of channel categories 
    extending $\mathcal{I}$ to all term graphs.
\end{definition}

\begin{example}\label{ex:semantics-example}
Let $B := \{ \mathtt{D}, \mathtt{F}, \mathtt{S} \}$ be the set of basic attack step labels from \Cref{ex:at-signature-example}, 
let $t\colon B \rightarrow \{0,1\}$ 
be the assignment that maps each basic attack step label to a 
truth value, 
    $$ \mathtt{D} \mapsto 1, \;  \mathtt{F} \mapsto 0, \;\mathtt{S} \mapsto 1, $$ 
and let $\mathcal{I}_t$ be the Boolean function 
interpretation from \Cref{ex:boolean-functions-interpretation}.
What are the semantics $\sem{T}{\mathcal{I}_t}$ of the attack tree $T$ from \Cref{ex:decomposition-example} under $\mathcal{I}_t$?

Since part of the defining property of the semantics under an interpretation is that it preserves compositional structure
(including copy gates), 
we first use the decomposition \eqref{eq:shorter-decomposition} from \Cref{ex:decomposition-example}, 
and then plug in the definitions:
\begin{equation*}
\begin{aligned}
\sem{T}{\mathcal{I}_t} 
&= \sem{\AND \circ (\OR \otimes \OR) \circ (\mathtt{D} \otimes \copymap \otimes \mathtt{S}) \circ \mathtt{F}}{\mathcal{I}_t} \\
&= \sem{\AND}{\mathcal{I}_t} \circ 
   (\sem{\OR}{\mathcal{I}_t} \otimes \sem{\OR}{\mathcal{I}_t}) \\
&\quad \circ (\sem{\mathtt{D}}{\mathcal{I}_t} \otimes \copymap \otimes \sem{\mathtt{S}}{\mathcal{I}_t}) 
   \circ \sem{\mathtt{F}}{\mathcal{I}_t} \\
&= \mathcal{I}_t(\AND) \circ 
   (\mathcal{I}_t(\OR) \otimes \mathcal{I}_t(\OR)) \\
&\quad \circ (\mathcal{I}_t(\mathtt{D}) \otimes \copymap \otimes \mathcal{I}_t(\mathtt{S})) 
   \circ \mathcal{I}_t(\mathtt{F}) \\
&= \mathcal{I}_t(\AND)(\mathcal{I}_t(\OR)(\mathcal{I}_t(\mathtt{D}),\: \mathcal{I}_t(\mathtt{F})), \\
&\qquad\qquad\;\; \mathcal{I}_t(\OR)(\mathcal{I}_t(\mathtt{F}),\: \mathcal{I}_t(\mathtt{S}))) \\
&= (1 \lor 0) \land (0 \lor 1) = 1.
\end{aligned}
\end{equation*}
Hence, under the given interpretation, 
$T$ evaluates to ``true'', as we would expect. 
This example shows once more that our notion of interpretation generalises the standard notion from propositional logic. 
\end{example}

The approach in this example works more generally: we can decompose an AT according to Theorem \ref{thm:decomposition-into-atomic-term-graphs}, interpret the resulting components using $\mathcal{I}$, and then put them back together using $\circ$ and $\otimes$ on the channel category side.

\begin{theorem} \label{thm:algorithm}
Let $T$ be an AT, and let $ T = \mathop{\bigcirc}\limits_{i = 1}^{N} \bigotimes_{j = 1}^{k_i} A_{i,j}$ be its decomposition into atomic AT components. Let $\mathcal{I}$ be an interpretation of $\mathsf{AT}(B)$ in a channel category $C$. Then
\[
\sem{T}{\mathcal{I}} = \mathop{\bigcirc}\limits_{i = 1}^{N} \bigotimes_{j = 1}^{k_i} \sem{A_{i,j}}{\mathcal{I}}.
\]
\end{theorem}

Effectively, this is a metric computation algorithm, provided that $\circ, \otimes$ are computationally tractable; our decomposition result follows from \cite{corradini1999algebraic}, which has a constructive proof. This approach is akin to Bayesian network analysis via string diagrams \cite{jacobs2018channel}, which maps the graphical structure of a Bayesian network to the corresponding matrix operations. In Section \ref{sec:compositionality-of-attack-tree-metrics}, we discuss how this result can be used to compute specific metrics.

Note that in an interpretation $\mathcal{I}$, there needs to be no connection between $\mathcal{I}(\mathtt{AND}_i)$ and $\mathcal{I}(\mathtt{AND}_j)$, and the operators $\mathcal{I}(\mathtt{AND}_i)$ need not be `commutative' in the sense that $\mathcal{I}(\mathtt{AND}_{i+j}) \circ \mathsf{swap}_{i,j} = \mathcal{I}(\mathtt{AND}_{i+j})$. One could demand such restrictions, but we leave these out for now to be able to apply our framework to extensions of ATs, which typically have non-commutative structure.

\section{Application to existing metric frameworks}\label{sec:compositionality-of-attack-tree-metrics}

Up to this point, we have provided a compositional theory of 
the syntax and semantics of attack trees on a general level.
What we still need to demonstrate is that all common attack tree metrics
can indeed be viewed as instances of these general compositional semantics.
To this end, we show that all these attack tree metrics correspond to 
the semantics of attack trees under a certain interpretation in a channel category.
More concretely, we give channel category interpretations for the following metrics and semantics:
\begin{itemize}
\item The bottom-up semiring metrics of \cite{mauw2005foundations};
\item The propositional semiring metrics of \cite{lopuhaa2022efficient};
\item Fault tree unreliability \cite{ruijters2015fault};
\item The multiset semantics of \cite{mauw2005foundations};
\item The propositional semantics of \cite{lopuhaa2022efficient}.
\end{itemize}
Together, these cover most metrics used in the literature (with only few exceptions, as discussed in \Cref{sec:non-examples}); note that the set-semantics metrics of \cite{kordy2014dag} coincide with the propositional metrics when the underlying semiring is absorbing (see below), which is almost always the case \cite{lopuhaa2022efficient}.

Within both the bottom-up and propositional metric frameworks, there exist many relevant metrics, such as attack time, cost, skill, probability, etc. Both frameworks encode these by a set $R$ of possible metric values, and two binary operators $+$ and $\cdot$, corresponding to OR and AND gates, respectively. Each basic attack step label $b$ is assigned a metric value $\alpha(b) \in R$, and the map $\alpha\colon B \rightarrow R$ is called an \emph{attribution}. How these together define the AT's metric value differs per framework, as discussed in Section \ref{sec:example}. Nevertheless, both methodologies require the tuple $(R,+,\cdot)$ to satisfy the algebraic property that it forms a semiring:

\begin{definition}[Semiring]\label{def:semiring} 
    A \textit{semiring} is a set $R$ equipped with two binary operations $+$ and $\cdot$, 
    and two elements $0$ and $1$,
    such that $+$ and $\cdot$ are binary associative commutative operations on $R$ with unit $0$ and $1$, respectively, and such that $\cdot$ distributes over $+$, i.e.
    \begin{align*}
        r \cdot (s + t) &= r \cdot s + r \cdot t
    \end{align*}
    for all $r, s, t \in R$. If furthermore $r+(r\cdot s) = r$ for all $r,s$, we call $R$ \emph{absorbing}.
\end{definition}

An overview of metrics from the literature is given in Table \ref{table:semirings-overview}. All are absorbing except \emph{max challenge} and fault tree \emph{unreliability}.

\begin{table}
    \caption{Overview of semirings $(R, +, \,\cdot\,,\, 0, 1)$.}\label{table:semirings-overview}
    \centering
    \begin{tabular}{@{}lccccc@{}}
    \toprule
    Type of metric                                                         & $R$          & $+$    & $\cdot$ & $0$      & $1$      \\ \midrule
    Min. cost                                                              & $[0,\infty]$ & $\min$ & $+$     & $\infty$ & $0$      \\
    \begin{tabular}[c]{@{}l@{}}Min. time\\ (parallel)\end{tabular}         & $[0,\infty]$ & $\min$ & $\max$  & $\infty$ & $0$      \\
    \begin{tabular}[c]{@{}l@{}}Min. time \\ (sequential)\end{tabular}      & $[0,\infty]$ & $\min$ & $+$     & $\infty$ & $0$      \\
    Max. challenge                                                         & $[0,\infty]$ & $\max$ & $\max$  & $\infty$ & $\infty$ \\
    Max. probability                                                       & $[0,1]$      & $\max$ & $\cdot$ & $0$      & $1$      \\ \midrule
    \begin{tabular}[c]{@{}l@{}}Unreliability \\ (fault trees)\end{tabular} & $[0,\infty)$ & $+$    & $\cdot$ & $0$      & $1$      \\ \bottomrule
    \end{tabular}
\end{table}

\subsection{Bottom-up semiring metrics}

We begin with the case of bottom-up metrics.

\begin{definition}
    Let $B$ be a set, let $R$ be a semiring and let $\alpha: B\to R$ be an attribution.
    The \emph{bottom-up interpretation} $\mathcal{I}_{\mathsf{bu}}^\alpha$ is defined as follows:
    \begin{align*}
        \mathcal{I}_{\mathsf{bu}}^\alpha \colon \mathsf{AT}(B) &\to \mathsf{Functions}_R, \\
        \mathcal{I}_{\mathsf{bu}}^\alpha(\AND_i)(x_1, \dots, x_i) &:= \prod_{k=1}^i x_k, \\
        \mathcal{I}_{\mathsf{bu}}^\alpha(\OR_i)(x_1, \dots, x_i) &:= \sum_{k=1}^i x_k, \\
        \mathcal{I}_{\mathsf{bu}}^\alpha(b)(\,) &:= \alpha(b). \qquad (b \in B)
    \end{align*}
\end{definition}

The reason we call this the \emph{bottom-up} interpretation is that 
the semantics under this interpretation recover exactly the 
\emph{bottom-up} semiring metrics of \cite{mauw2005foundations}.
Their name comes from the following application of Theorem \ref{thm:algorithm},
which essentially states that the semantics of an an attack tree $T$ under $\mathcal{I}_{\mathsf{bu}}^\alpha$ can be computed 
recursively (i.e.~``bottom-up'') on the structure of $T$.

\begin{theorem} \label{thm:bu}
    Let $B$ be a set, let $R$ be a semiring and let $\alpha: B\to R$ be an attribution.
    Moreover, let $T$ be an attack tree over $B$ whose unique output $R_T$ has $n$ children, 
    and let $T'_i$ be the attack tree rooted at the $i$-th child of $R_T$, for all $i\in \{1, \dots, n\}$.
    Then:
        $$ \sem{T}{\mathcal{I}_{\mathsf{bu}}^\alpha} = 
        \begin{cases}
            \alpha(b) & \text{if } \labelfun(R_T) = b, b\in B,  \\
            \sum_{i=1}^{n} \sem{T'_i}{\mathcal{I}_{\mathsf{bu}}^\alpha} & \text{if } \labelfun(R_T) = \OR_n,  \\
            \prod_{j=1}^{n} \sem{T'_i}{\mathcal{I}_{\mathsf{bu}}^\alpha} & \text{if } \labelfun(R_T) = \AND_n.
        \end{cases}$$
\end{theorem}
\begin{proof}
    When $T$ is tree-structured (i.e.~every node has a unique parent node), 
    the claim follows directly by induction on the tree structure of $T$,  applying Theorem \ref{thm:algorithm} successively.
    For the general case, we may rewrite $T$ into a tree-structured attack tree 
    in a way that preserves the semantics under $\mathcal{I}_{\mathsf{bu}}^\alpha$, since 
    the interpretation of all function symbols is \emph{deterministic},
    i.e.~all function symbols $f$ satisfy \\
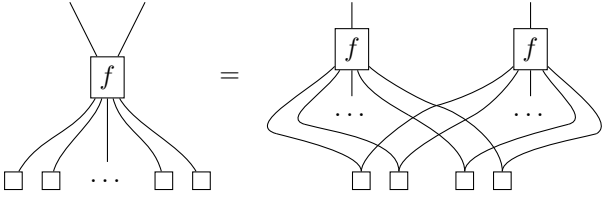
    \begin{equation*}
        \begin{tikzpicture}
	\begin{pgfonlayer}{nodelayer}
		\node [style={blank_rectangle}] (52) at (-3.5, 0) {$f$};
		\node [style=none] (53) at (-4.5, 2) {};
		\node [style={blank_rectangle}] (58) at (-6, -2.75) {};
		\node [style={blank_rectangle}] (59) at (-5, -2.75) {};
		\node [style={blank_rectangle}] (60) at (-2, -2.75) {};
		\node [style={blank_rectangle}] (61) at (-1, -2.75) {};
		\node [style=none] (62) at (-3.5, -2.75) {$\dots$};
		\node [style=none] (63) at (-3.5, -2.25) {};
		\node [style=none] (64) at (-2.5, 2) {};
		\node [style={blank_rectangle}] (65) at (3, 0.75) {$f$};
		\node [style=none] (66) at (3, 2) {};
		\node [style=none] (67) at (1, -1.5) {};
		\node [style=none] (68) at (1.75, -1.25) {};
		\node [style=none] (69) at (4.5, -1) {};
		\node [style=none] (70) at (5.25, -0.75) {};
		\node [style=none] (71) at (3, -1) {$\dots$};
		\node [style=none] (72) at (3, -0.5) {};
		\node [style={blank_rectangle}] (74) at (7.75, 0.75) {$f$};
		\node [style=none] (75) at (7.75, 2) {};
		\node [style=none] (76) at (5.25, -0.75) {};
		\node [style=none] (77) at (6.25, -1) {};
		\node [style=none] (78) at (8.75, -1.25) {};
		\node [style=none] (79) at (9.5, -1.5) {};
		\node [style=none] (80) at (7.75, -1) {$\dots$};
		\node [style=none] (81) at (7.75, -0.5) {};
		\node [style={blank_rectangle}] (82) at (3.25, -2.75) {};
		\node [style={blank_rectangle}] (83) at (4.25, -2.75) {};
		\node [style={blank_rectangle}] (84) at (6, -2.75) {};
		\node [style={blank_rectangle}] (85) at (7, -2.75) {};
		\node [style=none] (86) at (-0.25, 0) {$=$};
		\node [style=none] (87) at (-3.5, -3.5) {};
		\node [style=none] (88) at (5.25, -3.5) {};
	\end{pgfonlayer}
	\begin{pgfonlayer}{edgelayer}
		\draw (53.center) to (52);
		\draw [in=-120, out=60] (58) to (52);
		\draw [in=-105, out=75, looseness=0.75] (59) to (52);
		\draw [in=-75, out=105] (60) to (52);
		\draw [in=-60, out=120, looseness=0.75] (61) to (52);
		\draw (63.center) to (52);
		\draw (52) to (64.center);
		\draw (66.center) to (65);
		\draw [in=-135, out=165] (67.center) to (65);
		\draw [in=-120, out=165] (68.center) to (65);
		\draw [in=-75, out=150, looseness=0.75] (69.center) to (65);
		\draw [in=-45, out=165, looseness=0.50] (70.center) to (65);
		\draw (72.center) to (65);
		\draw (75.center) to (74);
		\draw [in=-135, out=15, looseness=0.75] (76.center) to (74);
		\draw [in=-120, out=30, looseness=0.75] (77.center) to (74);
		\draw [in=-60, out=15] (78.center) to (74);
		\draw [in=-45, out=30, looseness=0.75] (79.center) to (74);
		\draw (81.center) to (74);
		\draw [in=345, out=90, looseness=0.50] (82) to (67.center);
		\draw [in=345, out=90, looseness=0.75] (83) to (68.center);
		\draw [in=-30, out=90, looseness=0.75] (84) to (69.center);
		\draw [in=345, out=90, looseness=0.75] (85) to (70.center);
		\draw [in=195, out=90, looseness=0.50] (82) to (76.center);
		\draw [in=-150, out=90, looseness=0.50] (83) to (77.center);
		\draw [in=195, out=90, looseness=0.50] (84) to (78.center);
		\draw [in=195, out=90, looseness=0.50] (85) to (79.center);
	\end{pgfonlayer}
\end{tikzpicture}
    \end{equation*}
    in the semantics under $\mathcal{I}_{\mathsf{bu}}^\alpha$.
\end{proof}

\begin{example}
In the \emph{min cost} metric, semiring addition is $\min$, and semiring multiplication is $+$. Hence, as in Section \ref{sec:example}, an AND-gate adds the costs of its children, since all need to be activated; and an OR-gate selects the minimum cost of its children, as only one needs to be activated.
\end{example}

\subsection{The stochastic matrix interpretation of an attack tree}

In order to treat the case of propositional semiring metrics 
and fault tree unreliability in a uniform way, 
we first introduce the \emph{stochastic matrix interpretation}. 
This makes use of the notion of a \emph{Boolean-indexed stochastic matrix} over a semiring.
Intuitively, a Boolean-indexed stochastic matrix can be understood as a (generalised) conditional probability table, 
describing the probabilities (or cost, time, etc.) of transitioning from one vector of Booleans to another, possibly of a different size.

\begin{definition}
    Let $R$ be a semiring. 
    A \emph{Boolean-indexed matrix} over $R$ is a $2^j\times 2^i$-matrix over $R$ for some $i, j\in \mathbb{Z}_{\geq 0}$. 
    Here, we write $2=\{0,1\}$ for the set of Booleans, 
    and thus view such matrices as indexed over vectors of Booleans.
    In particular, vectors $v\in R^2$ will be indexed as $(v_0,v_1)^\top$,
    and hence, \textbf{in the following, $v_1$ will always denote the \emph{second} component of $v$.}

    A $2^j\times 2^i$-matrix $A$ over $R$ is \emph{stochastic} if for all $x\in 2^i$,
    \begin{equation}\label{eq:stochastic}
        \sum_{y\in 2^j} A_{yx} = 1.
    \end{equation}
\end{definition}

When $R=([0,\infty),+,\cdot,0,1)$ is the semiring of nonnegative real numbers with their ordinary addition and multiplication, and $A$ is square, 
this recovers the usual notion of a column-stochastic matrix. 

The entry $A_{yx} \in R$ can be thought of as the probability of -- or cost, or time, etc., required to -- transition from the ``state'' $x$ to the ``state'' $y$, with both of these ``states`` being vectors of Booleans of possibly different sizes.

Whenever we write a Boolean-indexed stochastic matrix $P$ in tabular form, as in
\[
P = \begin{pmatrix}
    0.7 & 0.1 & 0.2 & 0.6 \\
    0.3 & 0.9 & 0.8 & 0.4
\end{pmatrix} \in [0,\infty)^{2^1\times2^2},
\]
we order the indices of $P$ according to their binary-integer value. For example, the entry $P_{(1),\,(1\,0)} = 0.8$ is the entry in the second row and third column, and represents the probability that the input $(0 \;1)$ transitions into the output $1$. 

The stochastic matrix interpretation now takes values in the following channel category.

\begin{definition}[$\mathsf{BoolStoch}_R$]\label{def:bool-stoch}
    Let $R$ be a semiring. 
    The channel category $\mathsf{BoolStoch}_R$ of Boolean-indexed stochastic matrices over $R$ is given as follows.
    For each $i,j\in \mathbb{Z}_{\geq 0}$, the set of channels $\mathsf{BoolStoch}_R(i,j)$ is the set of stochastic $2^j\times 2^i$-matrices over $R$.
    The sequential composition of stochastic matrices is given by matrix multiplication,
    while the parallel composition is given by the Kronecker product of matrices.
    The identity channel $\mathsf{id}_i$ is the $2^i\times 2^i$-identity matrix,
    and the copy and delete gates are given by the following matrices:
        $$ 
        \copymap := \begin{pmatrix}
            1 & 0 \\
            0 & 0 \\ 
            0 & 0 \\
            0 & 1
        \end{pmatrix}, \quad
        \delmap := \begin{pmatrix}
            1 & 1 
        \end{pmatrix}.
        $$
    Finally, the swap gate $\mathsf{swap}_{i,j}$ is given by the permutation matrix that swaps the first $i$ coordinates with the last $j$ ones.
\end{definition}

\begin{remark} The requirement \eqref{eq:stochastic} of working with stochastic matrices ensures that $\mathsf{BoolStoch}_R$, 
like $\mathsf{Functions}_R$, 
is a \emph{Markov category} \cite{fritz2020synthetic}.
This means that the only morphism with no output is
the $\mathsf{del}$ channel (``discard'').
\end{remark}

\begin{definition}[Stochastic matrix interpretation]
    Let $B$ be a set, let $R$ be a semiring, and let $\alpha = (\alpha_1, \alpha_0)$ 
    be a pair of functions $B \to R$ such that $\alpha_0(b) + \alpha_1(b) = 1$ for all $b\in B$.
    The \emph{stochastic matrix interpretation} $\mathcal{I}_{\mathsf{stoch}}^\alpha$ is defined as follows:
        $$ \mathcal{I}_{\mathsf{stoch}}^\alpha\colon \mathsf{AT}(B) \to \mathsf{BoolStoch}_R, $$
        $$ \mathcal{I}_{\mathsf{stoch}}^\alpha(\AND_i) := \begin{pmatrix}
            1 & 1 & \cdots & 1 & 0 \\
            0 & 0 & \cdots & 0 & 1 
        \end{pmatrix} \in R^{2 \times 2^i}, $$
        $$ \mathcal{I}_{\mathsf{stoch}}^\alpha(\OR_i) := \begin{pmatrix}
            1 & 0 & \cdots & 0 & 0 \\
            0 & 1 & \cdots & 1 & 1
        \end{pmatrix} \in R^{2 \times 2^i}, $$
        $$ \mathcal{I}_{\mathsf{stoch}}^\alpha(b) := \begin{pmatrix}
            \alpha_0(b) \\
            \alpha_1(b)
        \end{pmatrix}. $$
\end{definition}

To relate this interpretation to propositional semiring metrics and fault tree unreliability,
we first need to introduce some terms used to define the latter.

First, a \emph{basic attack step} in an attack tree component $T$ over a set $B$ of basic attack step labels is simply any node labelled with an element of $B$.

\begin{definition}[Basic attack step]
    Let $T$ be an attack tree component over a set $B$. 
    A basic attack step is a node $n\in N^T$ such that 
    $\labelfun^T(n)\in B$. 
    We write $\mathsf{BAS}(T)$ for the set of basic attack steps of $T$.
\end{definition}

Next, an \emph{attack} indicates for each basic attack step whether it is performed or not.

\begin{definition}[Attack]
    Let $T$ be an attack tree component.
    An \emph{attack} is a function $a\colon \mathsf{BAS}(T) \to \{0,1\}$ from the set of basic attack steps to the set of Booleans.
    We write $\mathsf{Attacks}(T)$ for the set of all attacks.
\end{definition}

Now, the \emph{structure function} of an attack tree determines for each attack whether it is successful or not. 
More generally, the structure function of an attack tree component determines the outcome of each output, for each attack and each additional input.

\begin{definition}[Structure function]
    Let $T: i \to j$ be an attack tree component.
    The \emph{structure function} of $T$ is the function 
    \begin{align*}
        &S_T\colon \mathsf{Attacks}(T) \times \{0,1\}^i \to \{0,1\}^j, \\ 
        &\;\;S_T(a, x) := \left(S_T(a, x;\,\outputs^T_k)\right)_{k=1}^j,
    \end{align*}
    where for all $n\in N^T$, $S_T(a, x;\,n)$ is defined recursively as 
    \[
        S_T(a, x;\,n) := \begin{cases}
            x_k   \quad\quad\quad\qquad \text{if } n = \inputs_k^T, \\
            a(n)  \quad\quad\quad\,\quad \text{if } n \in \mathsf{BAS}(T), \\
            \bigwedge_{m \in \ch^T(n)} S_T(a, x;\,m) \\ 
            \quad\quad\quad\qquad \,\quad \text{if } \labelfun^T(n) = \AND, \\
            \bigvee_{m \in \ch^T(n)} S_T(a, x;\,m) \\
            \quad\quad\quad\qquad \,\quad \text{if } \labelfun^T(n) = \OR.
        \end{cases}
    \]
\end{definition}

\noindent When $T:0 \to 1$ is an attack tree, it has no input nodes, and the structure function is a simply function 
    $$S_T: \mathsf{Attacks}(T) \to \{0,1\}.$$

\begin{example}
Consider the AT $T$ of Example \ref{ex:decomposition-example}. Let $n_{\mathtt{D}}$, $n_{\mathtt{F}}$, $n_{\mathtt{S}}$ be the nodes with labels $\mathtt{D}$, $\mathtt{F}$, $\mathtt{S}$, respectively. An attack on $T$ is a map $\{n_{\mathtt{D}}, n_{\mathtt{F}}, n_{\mathtt{S}}\} \rightarrow \{0,1\}$, which we may view as a three-dimensional binary vector $\vec{x} = (x_{\mathtt{D}}, x_{\mathtt{F}}, x_{\mathtt{S}})$.
In Example \ref{ex:semantics-example}, we saw that $S_T(101) = 1$. Using the same decomposition, but keeping the basic attack steps as variables, we find
\[
S_T(\vec{x}) = (x_{\mathtt{D}} \lor x_{\mathtt{F}}) \land (x_{\mathtt{D}} \lor x_{\mathtt{S}}).
\]
\end{example}

To keep the subsequent formulas concise, we will need one additional notational convenience.

\begin{notation}
    Let $T$ be an attack tree component over a set $B$ and let $\alpha: B \to R$ be an attribution (i.e.~a function from $B$ to some semiring $R$). We define 
        $$ \overline{\alpha}: \mathsf{BAS}(T) \to R, \quad \overline{\alpha}_i(v) := \alpha_i(\labelfun^T(v)).$$
\end{notation}

We are now ready to state an explicit formula for the stochastic matrix interpretation in terms of the notions introduced above. We will use this result to show that both the propositional semiring metrics of \cite{lopuhaa2022efficient} and the fault tree unreliability \cite{ruijters2015fault} can be viewed as semantics in a channel category. 
A proof is given in \Cref{sec:proof-matrix-interpretation-formula}.

\begin{theorem}\label{thm:matrix-semantics-formula}
    Let $T:i \to j$ be an attack tree component over a set $B$.
    Moreover, let $R$ be a semiring, and let $\alpha = (\alpha_0, \alpha_1)$ 
    be a pair of two functions $B \to R$ such that $\alpha_0(b) + \alpha_1(b) = 1$ for all $b\in B$.
    Then
    \[
    \sem{T}{\mathcal{I}_{\mathsf{stoch}}^\alpha} = (M^T_{yx})_{y\in 2^j, \,x\in 2^i},
    \]
    where
    \[
    M^T_{yx} = \!\!\sum_{\substack{a\in \mathsf{Attacks}(T) \\S_T(a, x) = y}} 
    \left(\prod_{ a(v) = 1} \overline{\alpha}_1(v)\right) \cdot 
    \left(\prod_{a(v) = 0} \overline{\alpha}_0(v)\right).
    \]
\end{theorem}

When $T$ is an attack tree, the statement of \Cref{thm:matrix-semantics-formula} can be simplified, using the following notation for the set of all successful attacks.

\begin{definition}[Successful attacks]
    Let $T$ be an attack tree.
    The set of all successful attacks is denoted by 
        $$ \mathsf{Suc}_T := S_T^{-1}(\{1\}) = \{a \in \mathsf{Attacks}(T) \mid S_T(a) = 1\}. $$
\end{definition}

The following is now a direct consequence of \Cref{thm:matrix-semantics-formula}.

\begin{corollary}\label{thm:main-result}
    Let $T$ be an attack tree over a set $B$. As before, let $R$ be a semiring, and let $\alpha = (\alpha_1, \alpha_0)$ 
    be a pair of two functions $B \to R$ such that $\alpha_0(b) + \alpha_1(b) = 1$ for all $b\in B$.
    Then 
        $$ \left(\sem{T}{\mathcal{I}_{\mathsf{stoch}}^\alpha}\right)_1 = \sum_{a\in \mathsf{Suc}_T} \left(\prod_{a(v) = 1} \overline{\alpha}_1(v)\right) \cdot \left(\prod_{a(v) = 0} \overline{\alpha}_0(v)\right) $$
    where the left-hand side is the \emph{second} component of the stochastic matrix semantics of $T$ (a vector in $R^2$).
\end{corollary}

\begin{example}
Consider again the minimal cost semiring $([0,\infty],\min,+,\infty,0)$. The restriction on $\alpha$ is now that $\min(\alpha_0(b),\alpha_1(b)) = 0$ for all $b$. We take $\alpha_0(b) = 0$, and $\alpha_1(b)$ to be the actual cost of any basic attack step labelled $b$: then \Cref{thm:main-result} tells us that
\begin{align*}
\left(\sem{T}{\mathcal{I}_{\mathsf{stoch}}^\alpha}\right)_1 &= \min_{a\in \mathsf{Suc}_T} \left(\sum_{a(v) = 1} \overline{\alpha}_1(v) + \sum_{a(v) = 0} \overline{\alpha}_0(v)\right)\\
&= \min_{a\in \mathsf{Suc}_T} \sum_{a(v) = 1} \overline{\alpha}_1(v).
\end{align*}
This is indeed \emph{min cost} as defined in \cite{lopuhaa2022efficient}. The next subsection applies this approach to general semiring metrics.
\end{example}

\subsection{Propositional semantics and metrics}

Using \Cref{thm:main-result},
we now show that the propositional semiring metrics of \cite{lopuhaa2022efficient}
are given by the semantics under a certain interpretation,
the \emph{propositional interpretation}. 
Hence, in particular, these semantics extend to arbitrary attack tree components,
preserving the compositional structure of attack trees. In this section, we assume all semirings are \emph{absorbing}.

The propositional interpretation is now defined as follows:

\begin{definition}
    Let $B$ be a set, let $R$ be an absorbing semiring,
    and let $\alpha\colon B\to R$ be an attribution.
    The \emph{($\alpha$-weighted) propositional interpretation} $\mathcal{I}_{\mathsf{prop}}^\alpha$ is
        $$ \mathcal{I}_{\mathsf{prop}}^\alpha := \mathcal{I}_{\mathsf{stoch}}^{\alpha'}, $$
    where $\alpha' = (\alpha_0', \alpha_1')$ is the pair of functions $B\to R$ given by:
        $$ \alpha_0'(b) := \alpha(b), \;\; \alpha_1'(b) := 1. \qquad (b \in B)$$
\end{definition}

The definition of propositional semiring metrics depends
on the notion of \emph{minimal successful attacks} of an AT $T$.

\begin{definition}\label{def:minimal-successful-attack-suites}
    Let $T$ be an attack tree over a set $B$.
    Identifying attacks $B\to \{0,1\}$ with subsets of $B$, the set of \emph{minimal successful attacks} of $T$ is defined as
        $$ \mathsf{MinSuc}_T := \{a\in \mathsf{Suc}_T \mid \not\exists a' \in \mathsf{Suc}_T: a' \subsetneq a \}. $$
    In other words, $\mathsf{MinSuc}_T$ is the set of all successful attacks that are minimal in the sense that 
    no other successful attack is a subset of them.
\end{definition}

As we will show later, the set of minimal successful attacks can also be obtained as the semantics of $T$ under a certain interpretation.

The following corollary of \Cref{thm:main-result} provides an explicit expression for the semantics of an attack tree $T$ under the propositional interpretation $\mathcal{I}_{\mathsf{prop}}^\alpha$,
coinciding precisely with the definition of propositional semiring metrics in \cite{lopuhaa2022efficient}.

\begin{corollary}\label{cor:prop-semring-metrics-formula}
    Let $B$ be a set, let $R$ be an absorbing semiring, 
    and let $\alpha: B\to R$ be an attribution.
    Then for any attack tree $T$, we have:
        $$ \left(\sem{T}{\mathcal{I}_{\mathsf{prop}}^\alpha}\right)_1 = \sum_{a \in \mathsf{MinSuc}_T} \prod_{a(v) = 1} \overline{\alpha}(v). $$
\end{corollary}
\begin{proof}
    By the definition of $\mathcal{I}_{\mathsf{prop}}^\alpha$ and \Cref{thm:main-result},
    \begin{equation}\label{eq:prop-semantics-formula}
        \left(\sem{T}{\mathcal{I}_{\mathsf{prop}}^\alpha}\right)_1 = \sum_{a\in \mathsf{Suc}_T} \left(\prod_{a(v) = 1} \overline{\alpha}_1(v)\right).
    \end{equation}
    Since $R$ is absorbing, adding the product of some superset $a' \supseteq a$ of $a\in \mathsf{Suc}_T$ does not change the sum, 
    and hence, replacing $\mathsf{Suc}_T$ by $\mathsf{MinSuc}_T$ in \Cref{eq:prop-semantics-formula} we obtain the desired identity.
\end{proof}

Theorem \ref{thm:algorithm} now gives an algorithm to compute propositional semiring metrics, through repeated multiplication and Kronecker multiplication of matrices. 
This is completely different from existing algorithms based on binary decision diagrams \cite{lopuhaa2022efficient} and clone deletion \cite{kordy2014dag}. 
Since the size of the involved matrices are exponential in the number of inputs/outputs, to implement this effectively the decomposition of Theorem \ref{thm:decomposition-into-atomic-term-graphs} should be chosen such that the number of inputs/outputs at each composition is minimal. 
Optimising this is beyond the scope of this paper.

\subsection{Qualitative attack tree semantics}

In addition to attack tree metrics, qualitative semantics have also been proposed.
We will discuss two such semantics, the \emph{propositional semantics} \cite{lopuhaa2022efficient}
given by the set of minimal attacks, and the \emph{multiset semantics} \cite{mauw2005foundations}.
These are simply referred to as ``the semantics of an attack tree'' in these works.
In contrast, we do not make a choice as to which qualitative semantics is the default, 
and also consider attack tree \emph{metrics} as a type of (quantitative) semantics.

\subsubsection{Minimal successful attacks}

The set of minimal successful attacks $\mathsf{MinSuc}_T$ of an attack tree $T$ (see \Cref{def:minimal-successful-attack-suites})
can be realised as the semantics of $T$ under the interpretation $\mathcal{I}_{\mathsf{prop}}^\alpha$, 
for a specific semiring $\mathsf{AC}(T)$ which we now construct.

For any attack tree $T$ over a set $B$, we consider elements $a\colon \mathsf{BAS}(T) \rightarrow \{0,1\}$ of $\mathsf{Attacks}(T)$ as subsets of $\mathsf{BAS}(T)$, by considering them as characteristic functions. We then let 
    $$\mathsf{AC}(T) := \{A\subseteq \mathsf{Attacks}(T) \mid \forall a,a'\in A: a\not\subseteq a' \land a' \not\subseteq a\} $$
be the set of antichains in $\mathsf{Attacks}(T)$.
The set $\mathsf{AC}(T)$ becomes a semiring under the following operations:
    $$ A + A' := (A\cup A')^{\mathsf{AC}}, \;\; A \cdot A' := (\{a \cup a' \mid a \in A, a' \in A'\})^{\mathsf{AC}}, $$
where
    $$ 
    S^{\mathsf{AC}} := \{a \in S \mid \forall a'\in S: a' \subsetneq a\} 
    $$
associates to $S$ its antichain of elements that are minimal with respect to $\subseteq$.
The additive and multiplicative units of $\mathsf{AC}(T)$ are given by the empty set $\emptyset$ and the singleton set $\{\emptyset\}$, respectively.
Moreover, $\mathsf{AC}(T)$ is absorbing.
(In fact, for finite $B$, $\mathsf{AC}(T)$ is the \emph{free distributive lattice} on $\mathsf{BAS}(T)$.) \\
The following theorem now shows that the set of minimal attacks for an attack tree $T$
is equivalently given by the semantics of $T$ 
under the propositional interpretation with respect to a certain attribution in $\mathsf{AC}(T)$. 
The proof is an application of \Cref{cor:prop-semring-metrics-formula} and given in \Cref{sec:proof-min-suc-semantics}.

\begin{theorem}\label{thm:min_suc_semantics}
    Let $T$ be an attack tree over a set $B$, 
    and assume that there is no distinction between basic attack steps and their labels, in the sense that $B=\mathsf{BAS}(T)$ and $\labelfun^T(v)=v$ for all $v\in \mathsf{BAS}(T)$. 
    Let 
        $$ \alpha\colon B \to \mathsf{AC}(T), \;\; b \mapsto \{\{b\}\}, $$
    be the attribution that assigns to each basic attack step $b$ the singleton antichain $\{\{b\}\}$.
    Then 
        $$ \mathsf{MinSuc}_T = \left(\sem{T}{\mathcal{I}_{\mathsf{prop}}^\alpha}\right)_1. $$
\end{theorem}

Note that the assumption that there is no distinction between basic attack steps and labels is not restrictive, 
as this is an inherent feature of how attack trees are formalised in context of the propositional semantics as treated in \cite{lopuhaa2022efficient}.

\subsubsection{The multiset of successful attacks}

A different approach to qualitative semantics of attack trees are the \emph{multiset semantics},
which are simply called ``the semantics of an attack tree'' in \cite{mauw2005foundations}.
By definition, these are the bottom-up semantics associated to the following semiring.

For any set $B$ let $\mathsf{M}(2^B)$ be the set of multisets of subsets of $B$.
The set $\mathsf{M}(2^B)$ becomes a semiring under the following operations:
    $$ A + A' := A \cup A', \;\; A \cdot A' := \{a \cup a' \mid a \in A, a' \in A\}. $$
where $\cup$ denotes the union of multisets, adding up multiplicities, 
and the multiset comprehension notation on the right-hand side similarly keeps track of multiplicities.
The additive and multiplicative units of $\mathsf{M}(2^B)$ are given by the empty set $\emptyset$ and the singleton set $\{\emptyset\}$, respectively.

Under the multiset semantics, 
an attack step with multiple parents in the attack tree is interpreted as being performed repeatedly, 
once for each parent. 
In particular, in contrast to the semantics in terms of minimal attacks, 
sharing of attack steps in the attack tree does not affect these semantics. 

\subsection{Fault tree unreliability}

A further important consequence of \Cref{thm:main-result} is the case of fault tree unreliability, 
the probability that the system represented by the fault tree fails.

\begin{definition}
    Let $R := [0, \infty)$ be the semiring of nonnegative reals,
    with its usual addition and multiplication, and let $B$ be a set, 
    whose elements we think of as \emph{basic events}.
    Moreover, let $\alpha\colon B \to R$ be an attribution such that 
    $\alpha(b) \in [0,1]$ for all $b\in B$. 
    (This is usually called a \emph{probabilistic status vector} 
    in the context of fault trees.)
    The \emph{unreliability interpretation} is the following interpretation of the signature $\mathsf{FT}(B) := \mathsf{AT}(B)$ of fault trees over $B$:
    \begin{align*}
    \mathcal{I}_{\mathsf{unrel}}\colon \mathsf{FT}(B) &\to \mathsf{BoolStoch}_R, \\
    \mathcal{I}_{\mathsf{unrel}} &:= \mathcal{I}_{\mathsf{stoch}}^{\alpha'},
    \end{align*}
    where $\alpha' =(\alpha'_0, \alpha'_1)$ is the pair of functions $B\to R$:
        $$ \alpha'_0(b) = 1-\alpha(b), \;\; \alpha'_1(b)=\alpha(b). \qquad (b\in B)$$
\end{definition}
With this definition in place, we obtain:
\begin{corollary}
    Let $T$ be an attack tree (thought of as a fault tree) over some set $B$
    (``basic events''),
    and let $\alpha\colon B\to [0,1]$ (``probabilistic status vector'').
    Moreover, let $(X_v)_{v\in \mathsf{BAS}(T)}$ be independent random variables, 
    each Bernoulli-distributed with parameter $p_v = \overline{\alpha}(v)$.
    Then
        $$ \left(\sem{T}{\mathcal{I}_{\mathsf{unrel}}}\right)_1 = \mathbb{P}\left[S_T\left((X_v)_{v\in \mathsf{BAS}(T)}\right) = 1\right], $$
    where $S_T$ is the structure function of $T$. 
    In other words, the semantics of $T$ under the unreliability interpretation correspond to precisely to the probability that the system $T$ fails.
\end{corollary}
\begin{proof}
    According to \Cref{thm:main-result}, we have:
    $$ \left(\sem{T}{\mathcal{I}_{\mathsf{unrel}}}\right)_1 = \!\!\!\!\sum_{a\in \mathsf{Suc}_T} \!\!\left(\prod_{a(v)=1} \!p_v\right) \cdot \left(\prod_{a(v)=0} \!(1-p_v)\right). $$
    Since, by definition, $\mathsf{Suc}_T = S_T^{-1}(\{1\})$, the expression on the right hand side is precisely the probability 
    we need. 
\end{proof}
Applying Theorem \ref{thm:algorithm} to fault tree unreliability yields exactly the matrix-based analysis of Bayesian network of \cite{jacobs2018channel}; note that Bayesian networks generalise fault trees.

\subsection{Non-examples}\label{sec:non-examples} 
Up to this point, we have shown which attack tree metrics and qualitative semantics arise as special cases of \Cref{def:semantics}.
We now turn to the converse question: which metrics or qualitative semantics do \emph{not} naturally fit into our framework?
Since metrics that preserve the channel category structure of attack tree components must also preserve the coarser modular composition considered in \cite{lopuhaa2024attack}, the non-examples of \emph{operad metrics} given therein also apply in our case. 
A concrete example is given by the \emph{mean time to compromise} \cite{mcqueen2006time}; see \cite{lopuhaa2024attack} for details.

\section{Conclusion}

The core of our results can be summarised in one slogan: 
\emph{attack tree metrics are functors of channel categories}.
This characterisation allows
one to define and analyse attack tree metrics in a modular manner, 
by how they behave on the basic building blocks of attack trees. 
Moreover, it fruitfully identifies attack trees as string diagrams, 
thereby connecting them to the numerous other areas, see \cite{piedeleu2025introduction}, in which string diagrams
have been applied.
An interesting question for future research is how to effectively exploit the compositionality  
of attack tree metrics for algorithmic purposes in practice. 
Theorem \ref{thm:algorithm} is a first qualitative step in this direction.
However, more research is needed to design generic algorithms at this level that account for efficiency, and to evaluate their complexity and practical performance against specialised methods.

\bibliographystyle{alpha}
\bibliography{bibliography.bib}

@article{lopuhaa2022efficient,
  title={Efficient and generic algorithms for quantitative attack tree analysis},
  author={Lopuha{\"a}-Zwakenberg, Milan and Budde, Carlos E and Stoelinga, Mari{\"e}lle},
  journal={IEEE Transactions on Dependable and Secure Computing},
  volume={20},
  number={5},
  pages={4169--4187},
  year={2022},
  publisher={IEEE}
}

@article{aslansefat2020dynamic,
  title={Dynamic fault tree analysis: state-of-the-art in modeling, analysis, and tools},
  author={Aslansefat, Koorosh and Kabir, Sohag and Gheraibia, Youcef and Papadopoulos, Yiannis},
  journal={Reliability management and engineering},
  pages={73--112},
  year={2020},
  publisher={CRC Press}
}

@inproceedings{kumar2017quantitative,
  title={Quantitative security and safety analysis with attack-fault trees},
  author={Kumar, Rajesh and Stoelinga, Mari{\"e}lle},
  booktitle={2017 IEEE 18th International Symposium on High Assurance Systems Engineering (HASE)},
  pages={25--32},
  year={2017},
  organization={IEEE}
}

@article{ruijters2015fault,
  title={Fault tree analysis: A survey of the state-of-the-art in modeling, analysis and tools},
  author={Ruijters, Enno and Stoelinga, Mari{\"e}lle},
  journal={Computer science review},
  volume={15},
  pages={29--62},
  year={2015},
  publisher={Elsevier}
}

@inproceedings{kaiser2003new,
  title={A new component concept for fault trees},
  author={Kaiser, Bernhard and Liggesmeyer, Peter and M{\"a}ckel, Oliver},
  booktitle={Proceedings of the 8th Australian workshop on Safety critical systems and software-Volume 33},
  pages={37--46},
  year={2003}
}

@inproceedings{mauw2005foundations,
  title={Foundations of attack trees},
  author={Mauw, Sjouke and Oostdijk, Martijn},
  booktitle={International Conference on Information Security and Cryptology},
  pages={186--198},
  year={2005},
  organization={Springer}
}

@article{fritz2023free,
  title={Free gs-monoidal categories and free {M}arkov categories},
  author={Fritz, Tobias and Liang, Wendong},
  journal={Applied Categorical Structures},
  volume={31},
  number={2},
  pages={21},
  year={2023},
  publisher={Springer}
}

@incollection{mcqueen2006time,
  title={Time-to-compromise model for cyber risk reduction estimation},
  author={McQueen, Miles A and Boyer, Wayne F and Flynn, Mark A and Beitel, George A},
  booktitle={Quality of Protection: Security Measurements and Metrics},
  pages={49--64},
  year={2006},
  publisher={Springer}
}

@article{kordy2014dag,
  title={{DAG}-based attack and defense modeling: Don’t miss the forest for the attack trees},
  author={Kordy, Barbara and Pi{\`e}tre-Cambac{\'e}d{\`e}s, Ludovic and Schweitzer, Patrick},
  journal={Computer science review},
  volume={13},
  pages={1--38},
  year={2014},
  publisher={Elsevier}
}

@article{schneier1999attack,
  title={Attack trees},
  author={Schneier, Bruce},
  journal={Dr. Dobb’s journal},
  volume={24},
  number={12},
  pages={21--29},
  year={1999}
}

@article{gadducci1996algebraic,
  title={On the algebraic approach to concurrent term rewriting},
  author={Gadducci, Fabio},
  journal={Bulletin-European Association for Theoretical Computer Science},
  volume={59},
  pages={412--413},
  year={1996},
  publisher={EUROPEAN ASSOCIATION FOR THEORETICAL COMPUTER SCIENCE}
}

@article{jacobs2018channel,
  title={A channel-based exact inference algorithm for {B}ayesian networks},
  author={Jacobs, Bart},
  journal={arXiv preprint arXiv:1804.08032},
  year={2018}
}

@article{fritz2020synthetic,
  title={A synthetic approach to {M}arkov kernels, conditional independence and theorems on sufficient statistics},
  author={Fritz, Tobias},
  journal={Advances in Mathematics},
  volume={370},
  year={2020},
  publisher={Elsevier}
}

@article{cho2019disintegration,
  title={Disintegration and {B}ayesian inversion via string diagrams},
  author={Cho, Kenta and Jacobs, Bart},
  journal={Mathematical Structures in Computer Science},
  volume={29},
  number={7},
  pages={938--971},
  year={2019},
  publisher={Cambridge University Press}
}

@inproceedings{lopuhaa2024attack,
  title={Attack tree metrics are operad algebras},
  author={Lopuha{\"a}-Zwakenberg, Milan},
  booktitle={2024 IEEE 37th Computer Security Foundations Symposium (CSF)},
  pages={665--679},
  year={2024},
  organization={IEEE}
}

@article{fong2013causal,
  title={Causal theories: A categorical perspective on {B}ayesian networks},
  author={Fong, Brendan},
  journal={arXiv preprint arXiv:1301.6201},
  year={2013}
}

@inproceedings{bossuat2017evil,
  title={Evil Twins: Handling Repetitions in Attack--Defense Trees: A Survival Guide},
  author={Bossuat, Ang{\`e}le and Kordy, Barbara},
  booktitle={International Workshop on Graphical Models for Security},
  pages={17--37},
  year={2017},
  organization={Springer}
}

@article{corradini1999algebraic,
  title={An algebraic presentation of term graphs, via gs-monoidal categories},
  author={Corradini, Andrea and Gadducci, Fabio},
  journal={Applied Categorical Structures},
  volume={7},
  pages={299--331},
  year={1999},
  publisher={Springer}
}

@inproceedings{dong2017attack,
  title={An attack tree-based approach for vulnerability assessment of communication-based train control systems},
  author={Dong, Huiyu and Wang, Hongwei and Tang, Tao},
  booktitle={2017 Chinese Automation Congress (CAC)},
  pages={6407--6412},
  year={2017},
  organization={IEEE}
}

@inproceedings{beckers2014determining,
  title={Determining the probability of smart grid attacks by combining attack tree and attack graph analysis},
  author={Beckers, Kristian and Heisel, Maritta and Krautsevich, Leanid and Martinelli, Fabio and Meis, Rene and Yautsiukhin, Artsiom},
  booktitle={International Workshop on Smart Grid Security},
  pages={30--47},
  year={2014},
  organization={Springer}
}

@inproceedings{khand2007attack,
  title={An attack model development process for the cyber security of safety related nuclear digital {I}\&{C} systems},
  author={Khand, Parvaiz Ahmed and Seong, Poong Hyun},
  booktitle={Proceedings of the Korean Nucleary Society (KNS) Fall meeting},
  year={2007}
}

@inproceedings{coecke2018picturing,
  title={Picturing quantum processes: A first course on quantum theory and diagrammatic reasoning},
  author={Coecke, Bob and Kissinger, Aleks},
  booktitle={International conference on theory and application of diagrams},
  pages={28--31},
  year={2018},
  organization={Springer}
}

@article{biamonte2011categorical,
  title={Categorical tensor network states},
  author={Biamonte, Jacob D and Clark, Stephen R and Jaksch, Dieter},
  journal={AIP Advances},
  volume={1},
  number={4},
  year={2011},
  publisher={AIP Publishing}
}

@book{piedeleu2025introduction,
  title={An Introduction to String Diagrams for Computer Scientists},
  author={Piedeleu, Robin and Zanasi, Fabio},
  year={2025},
  publisher={Cambridge University Press}
}

@inproceedings{jhawar2015attack,
  title={Attack trees with sequential conjunction},
  author={Jhawar, Ravi and Kordy, Barbara and Mauw, Sjouke and Radomirovi{\'c}, Sa{\v{s}}a and Trujillo-Rasua, Rolando},
  booktitle={IFIP International Information Security and Privacy Conference},
  pages={339--353},
  year={2015},
  organization={Springer}
}

@inproceedings{kordy2010foundations,
  title={Foundations of attack--defense trees},
  author={Kordy, Barbara and Mauw, Sjouke and Radomirovi{\'c}, Sa{\v{s}}a and Schweitzer, Patrick},
  booktitle={International Workshop on Formal Aspects in Security and Trust},
  pages={80--95},
  year={2010},
  organization={Springer}
}

\appendix

\subsection{Proof of \texorpdfstring{\Cref{thm:matrix-semantics-formula}}{the explicit formula for the stochastic matrix interpretation}}\label{sec:proof-matrix-interpretation-formula}

By the uniqueness part of \Cref{thm:term-graphs-are-free-channel-category}, 
it is sufficient to prove the following two statements.

\begin{enumerate}
    \item The family of maps
    \[
        M^{(-)}: \termgraphs{\mathsf{AT}(B)(i, j)} \to \mathsf{BoolStoch}_R(i,j),
    \]
    \[
        T \mapsto \left(M^T_{yx}\right)_{y\in 2^j, \,x\in 2^i},
    \]
    that assign to each attack tree component $T:i\to j$ the matrix $M^T$ defined in \Cref{thm:matrix-semantics-formula}, is a functor of channel categories.
    \item The desired identity,
    \begin{equation}\label{eq:sem_eq_M}
        \sem{T}{\mathcal{I}_{\mathsf{stoch}}^\alpha} = M^{T},
    \end{equation}
    holds for every $T$ of the form $T = \langle f \rangle$, where $f\in \mathsf{AT}(B)$ is any function symbol.
\end{enumerate} 
\text{}

\par \emph{Step 1):} We we need to verify that $M^{(-)}$ preserves parallel and sequential compositions, and that it preserves the basic wiring channels: the identity, copy, delete, and swap gates.

\emph{Parallel compositions:} 
Let $T_1:i\to j$, $T_2:k \to l$ be attack tree components over $B$. 
For every attack $a\in \mathsf{Attacks}(T_1 \otimes T_2)$, 
 every $(x,z) \in 2^{i+k}$ and every $(y,w) \in 2^{j+l}$, 
$S_{T_1 \otimes T_2}(a, (x,z)) = (y, w)$ if and only if 
$S_{T_1}(a \circ \mathsf{inl}, x) = y$ and $S_{T_2}(a \circ \mathsf{inr}, z) = w$, where $\mathsf{inl}: \mathsf{BAS}(T_1) \to \mathsf{BAS}(T_1 \otimes T_2)$ and $\mathsf{inl}: \mathsf{BAS}(T_2) \to \mathsf{BAS}(T_1 \otimes T_2)$ are the inclusion maps from the basic attack steps of $T_1$ and $T_2$ to the basic attack steps of their parallel compositions.
Moreover, we have that
\[
\prod_{\substack{v \in \mathsf{BAS}(T_1 \otimes T_2) \\ a(v) = i}} \!\!\!\!\!\overline{\alpha}_i(v) = \left(\prod_{\substack{v \in \mathsf{BAS}(T_1) \\ (a \,\circ\, \mathsf{inl})(v) = i}} \!\!\!\!\overline{\alpha}_i(v)\right) \cdot \left(\prod_{\substack{v \in \mathsf{BAS}(T_2) \\ (a \,\circ\, \mathsf{inr})(v) = i}} \!\!\!\!\overline{\alpha}_i(v)\right),
\]
for all $i\in \{0,1\}$.
Therefore, may split the sum in the formula defining $M^{T_1 \otimes T_2}$ and redistribute the products therein to obtain:
\begin{align*}
    M^{T_1 \otimes T_2}_{(y,x), (w, z)} = M^{T_1}_{(y,x)}\cdot M^{T_2}_{(w, z)} = M^{T_1} \otimes M^{T_2}.
\end{align*}

\emph{Sequential compositions:} 
Let $T_1:i\to j$, $T_2:j \to k$ be attack tree components over $B$.
For every attack $a\in \mathsf{Attacks}(T_2 \circ T_1)$, 
and every $x\in 2^i$, $z\in 2^k$, 
$S_{T_2 \circ T_1}(a, x)=z$ if and only if 
there exists a $y\in 2^j$ such that 
$S_{T_1}(a \circ \mathsf{inr}, x) = y$ and  $S_{T_2}(a \circ \mathsf{inl}, y)=z$,
where, similar to before, $\mathsf{inl}$ and $\mathsf{inr}$ are the inclusion maps from $T_1$ and $T_2$ to $T_2 \circ T_1$.
To simplify notation, write
\[
\Pi(a) := \left(\prod_{a(v) = 1} \overline{\alpha}_1(v)\right) \cdot 
    \left(\prod_{a(v) = 0} \overline{\alpha}_0(v)\right),
\]
for any attack tree component $T$ and attack $a\in \mathsf{Attacks}(T)$.
Similar to the case of parallel compositions, we then have
\[
\Pi(a) = \Pi(a \circ \mathsf{inl}) \cdot\Pi(a \circ \mathsf{inr}).
\]
Putting all of these observations together, we calculate
\begin{align*}
    M^{T_2\circ T_1}_{zx} 
    &= \!\sum_{S_{T_2\circ T_1}(a, x)=z} \!\!\!\Pi(a) \\
    &= \sum_{y\in 2^j} \:\sum_{S_{T_1}(a_1,\, x) = y}\,
       \sum_{S_{T_2}(a_2,\, y) = z} \!\!\!\!\!\Pi(a_1)\Pi(a_2)\\
    &= \sum_{y\in 2^j} \left(\sum_{S_{T_1}(a_1,\, x) = y}\!\!\!\!\!\Pi(a_1)\right)\!\!
    \left(\sum_{S_{T_2}(a_2,\, y) = z}\!\!\!\!\!\Pi(a_2)\right)\\
    &= \sum_{y\in 2^j} M_{yx}^{T_1} M_{zy}^{T_2} \\
    &= \left(M^{T_2} M^{T_1}\right)_{zx}.
\end{align*}

\emph{Basic wiring channels:} Observe that if $T$ does not contain any basic attack step, then
the products over basic attack steps in the definition of $M^T$ are empty, and hence all summands appearing therein are $1$. 
Moreover, when there are no basic attack steps, there is only one attack (the empty attack), 
and therefore, there is also at most one summand.
Hence, in this case, $M^T_{yx}$ is $1$ if and only if $S_T(\,!\,, x)=y$ 
(where $!: \emptyset \to \{0,1\}$ is the empty attack). 
Since none of the basic wiring channels contain any basic attack steps, 
we can use this observation to see that they satisfy \Cref{eq:sem_eq_M}
by direct comparison. \\

\emph{Step 2):} 
It remains to show that \Cref{eq:sem_eq_M} holds for all atomic attack tree components associated to the function symbols in the signature $\mathsf{AT}(B)$. 
For the gates $\AND_i$ and $\OR_i$, as they do not contain any basic attack steps, this follows using the same observation as for the basic wiring channels.
Finally, let $b\in B$ be a basic attack step label. 
Then $M^{\langle b\rangle}$ is a $2^1\times 2^0$-matrix, 
which we can identify with a vector $\left(M^{\langle b\rangle}_y\right)$ in $R^2$.
Since for each $y\in \{0,1\}$, there is only one attack $a$ that satisfies 
$S_{\langle b\rangle}(a)=y$, 
the sum defining $M^{\langle b\rangle}$ reduces to
\[
    M^{\langle b\rangle}_{y} =  \begin{pmatrix}
        \overline{\alpha}_0(v) &  \overline{\alpha}_1(v)
    \end{pmatrix}^\intercal,
\]
where $v\in \mathsf{BAS}(\langle b\rangle)$ is the unique basic attack step in $\langle b\rangle$. 
Using that $\overline{\alpha}_i(v) = \alpha_i(\labelfun^T(v)) = \alpha_i(b)$ ($i\in \{0,1\}$), we obtain
\[
    M^{\langle b\rangle}_{y} =  \begin{pmatrix}
        \alpha_0(b) &  \alpha_1(b)
    \end{pmatrix}^\intercal = \sem{\langle b\rangle}{\mathcal{I}_{\mathsf{stoch}}^\alpha},
\]  
thus completing the proof. \qed

\subsection{Proof of \texorpdfstring{\Cref{thm:min_suc_semantics}}{the theorem on minimal successful attacks}}\label{sec:proof-min-suc-semantics}

By \Cref{cor:prop-semring-metrics-formula}, 
\begin{equation}\label{eq:big-union-min-suc}
\left(\sem{T}{\mathcal{I}_{\mathsf{prop}}^\alpha}\right)_1 = \left(\bigcup_{a \in \mathsf{MinSuc}_T} \prod_{a(b) = 1} \overline{\alpha}(b)\right)^{\mathsf{AC}},
\end{equation}
where the product is taken in the semiring $\mathsf{AC}(T)$. 
As before, we identify attacks $a\in \mathsf{MinSuc}_T$ with subsets of $B$, sometimes writing $b\in a$ instead of $a(b)=1$.
Then, using the definition of the product in $\mathsf{AC}(T)$, 
we obtain:
\begin{align*}
\prod_{a(b)=1}\overline{\alpha}(b)
&= \left\{
   \bigcup_{b\in a} \tilde{a}_b \;\Big|\;
   \forall b\in a:\;
   \tilde{a}_b \in \overline{\alpha}(b)
   \right\}^{\mathsf{AC}} \\[4pt]
&= \left\{
   \bigcup_{b\in a} \tilde{a}_b \;\Big|\;
   \forall b\in a:\;
   \tilde{a}_b \in \{\{b\}\}
   \right\}^{\mathsf{AC}} \\[4pt]
&= \left\{
   \bigcup_{b\in a} \{b\}
   \right\}^{\mathsf{AC}}
   = \{a\}^{\mathsf{AC}}
   = \{a\}.
\end{align*}
Therefore, the union over minimal successful attacks on the right-hand side of \Cref{eq:big-union-min-suc} is a union of singleton sets, giving, 
\[
\left(\sem{T}{\mathcal{I}_{\mathsf{prop}}^\alpha}\right)_1 = \left(\mathsf{MinSuc}_T\right)^{\mathsf{AC}} = \mathsf{MinSuc}_T,
\]
where the final equality follows because $\mathsf{MinSuc}_T$ is already an antichain. \qed 

\end{document}